%% file: aAMG4QCD.tex
\documentclass{siamltex}

\usepackage{color}
\usepackage{amsfonts,amsmath}
\usepackage{graphicx}
\usepackage{tikz}
\usepackage[algosection,vlined,ruled,linesnumbered]{algorithm2e}
\usepackage{subfigure}
\usepackage{booktabs}
\usepackage{wasysym}
\usepackage{multirow}

\newcommand{\sign}{{\rm sign}}
\newcommand{\Dslash}{\mathcal{D}}
\newcommand{\range}{\mathop{\mathrm{range}}\nolimits}

\newcommand{\MSAP}[1]{M_\mathit{SAP}^{(#1)}}
\newcommand{\MSAPone}{M_\mathit{SAP}}
\newcommand{\ESAP}{E_\mathit{SAP}}

\newcommand{\ninv}{n_\mathit{inv}}

\newcounter{confcounter}
\newcommand{\conflabel}[1]{\refstepcounter{confcounter} \label{#1}}

\DontPrintSemicolon
 
\title{Adaptive Aggregation Based Domain Decomposition Multigrid for the Lattice Wilson Dirac Operator  \thanks{This work was partially funded by Deutsche Forschungsgemeinschaft (DFG) Transregional Collaborative Research Centre 55 (SFB/TRR55)}}

\author{A. Frommer\thanks{Department of Mathematics, Bergische Universit\"at Wuppertal, 42097 Germany, {\tt \{frommer,kkahl,leder,rottmann\}@math.uni-wuppertal.de}.}
\and K. Kahl\footnotemark[2]
\and S. Krieg\thanks{Department of Physics, Bergische Universit\"at Wuppertal, 42097 Germany and J\"ulich Supercomputing Centre, Forschungszentrum J\"ulich, 52428 J\"ulich, Germany, {\tt s.krieg@fz-juelich.de}.}
\and B. Leder\footnotemark[2]
\and M. Rottmann\footnotemark[2]
}

\begin{document}
\maketitle
\begin{abstract}
In lattice QCD computations a substantial amount of work is spent in solving discretized versions of the Dirac equation. Conventional Krylov solvers show critical slowing down for large system sizes and physically interesting parameter regions. We present a domain decomposition adaptive algebraic multigrid method used as a preconditioner to solve the ``clover improved'' Wilson discretization of the Dirac equation. This approach combines and improves two approaches, namely domain decomposition and adaptive algebraic multigrid, that have been used separately in lattice QCD before. We show in extensive numerical tests conducted with a parallel production code implementation that considerable speed-up over conventional Krylov subspace methods, domain decomposition methods and other hierarchical approaches for realistic system sizes can be achieved.
\end{abstract}



\section{Introduction}
Lattice QCD simulations are among the world's most demanding computational problems, and a significant part of today's supercomputer resources is spent in these simulations \cite{PRACE:ScAnnRep12,PRACE:ScC12}. Our concern in this paper is three-fold: We want to make the mathematical modeling related with QCD and lattice QCD more popular in the scientific computing community and therefore spend some effort on explaining fundamentals. On top of that we develop a new and efficient adaptive algebraic multigrid method to solve systems with the discretized Dirac operator, and we show results for a large number of numerical experiments based on an advanced, production code quality implementation with up-to-date physical data.

The computational challenge in lattice QCD computations consists of repeatedly solving very large sparse linear systems
\begin{equation}\label{eq:discreteDirac}
  Dz = b,
\end{equation}
where $D = D(U,m)$ is a discretization, typically the Wilson discretization,  of the Dirac operator on a four-dimensional space-time lattice. The Wilson Dirac operator depends on a gauge field $U$ and a mass constant $m$. %
Recently, lattices with up to $144\times64^3$ lattice points have been used, involving the solution of linear systems with $452,\!984,\!832$ unknowns \cite{Alexandrou:2011db,Bae:2011ff,Bali:2012qs,Durr:2010aw,Fritzsch:2012wq}. Usually these linear systems are solved by standard Krylov subspace methods. Their iteration count increases tremendously when approaching the physically relevant parameter values (i.e., physical mass constants and lattice spacing $a\rightarrow 0$), a phenomenon referred to as ``critical slowing down'' in the physics literature. Thus it is of utmost importance to develop preconditioners 
which overcome these scaling issues. The most common 
preconditioners 
nowadays are odd-even preconditioning~\cite{Degrand1990211,Lippert19991357}, deflation%
~\cite{Luescher2007}, and domain decomposition%
~\cite{Nobile2012,Luescher2003}. While these approaches yield significant speed-ups over the unpreconditioned versions, their scaling behavior is unchanged and critical slowing down still occurs. 

Multigrid methods have been considered in the lattice QCD community as well, motivated by their potential (e.g., for elliptic PDEs) of convergence independent of the lattice spacing. However, due to the random nature of the gauge fields involved, the treatment of the lattice Dirac equation by 
multigrid methods based solely on the underlying PDE, has been elusive for the last twenty years~\cite{Ben-Av:1990gb,Brower:1987dd,Kalkreuter:1994fz,Vink:1991fa}. With the advent of \textit{adaptive algebraic} multigrid methods, effective preconditioners for QCD calculations could be constructed in recent years. The pioneering work from \cite{MGClark2009, MGClark2010_1, MGClark2007, MGClark2010_2} showed very promising results. There, an adaptive non-smoothed aggregation approach based on~\cite{Brezina2005} has been proposed for the solution of the Wilson Dirac system. An implementation is publicly available as part of the QOPQDP package~\cite{wwwQOPQDP}.

Within the physics community, another hierarchical technique, the recently proposed domain decomposition type solver named \textit{inexact deflation} developed in \cite{Luescher2007} is widely used. A well-optimized code for this solver is publicly available~\cite{wwwDDHMC}. Inexact deflation can be regarded as an adaptive method as well. It performs a setup phase which allows the construction of a smaller system, the \textit{little Dirac} operator, which is then used as part of an efficient preconditioner. Although there is an intimate connection with the aggregation based multigrid approach from \cite{Brezina2005}, inexact deflation seems to have been developed completely independently. As a consequence, the inexact deflation method does not resemble a typical multigrid method in the way its ingredients are arranged. In particular, it requires the little Dirac system to be solved to high accuracy in each iteration.

In this paper we present a multigrid method that combines aspects from~\cite{Luescher2007}, namely a domain decomposition smoother, and from non-smoothed aggregation as in \cite{MGClark2010_1, MGClark2010_2}. Our approach elaborates on the multigrid methods from \cite{MGClark2010_1, MGClark2010_2} in that we use a domain decomposition method as the smoother instead of the previously used Krylov subspace smoother. This allows for a natural and efficient parallelization, also on hybrid architectures. Moreover, we substantially improve the adaptive setup from \cite{MGClark2010_1, MGClark2010_2} and \cite{Luescher2007} in the sense that less time is required to compute the operator hierarchy needed for an efficient multigrid method. Our approach can also be regarded as turning the domain decomposition technique from \cite{Luescher2007} into a true multigrid method. The ``little Dirac'' system now needs to be solved only to low accuracy. This allows, in particular, to apply the method recursively and thus opens the way for a more efficient {\em multi}-grid method instead of just a two-grid method. With the inexact deflation approach this is not possible.

The paper is organized as follows. In section~\ref{qcd_section} we give an introduction into lattice QCD for the non-specialist and we introduce the domain decomposition Schwarz method in this context. In section~\ref{section:AMG} we first outline algebraic multigrid methods in general and then focus on aggregation based approaches. Thereby we address the peculiarities of lattice QCD and explain different possible adaptive strategies for the construction of the multigrid hierarchy. The inexact deflation method from \cite{Luescher2007} is discussed in section~\ref{sec:PID}, where in particular we point out the differences to a multigrid method and describe the adaptive nature of its setup. In section~\ref{sec:DDML} we finally formulate our multigrid approach, for which we present thorough numerical tests in section~\ref{sec:NR}.

\section{Lattice Quantum Chromodynamics} \label{qcd_section}
Quantum Chromodynamics (QCD) or the theory of strong interactions, is a quantum field theory in four-dimensional space-time and part of the standard model of elementary particle physics. It has a high predictive power, i.e., a small number of free parameters. Predictions that can be deduced from this theory are amongst others the masses of hadrons, composite particles bound by the strong interaction (e.g., nucleon, pion; cf.~\cite{Durr21112008}). The masses of hadrons and many other predictions have to be obtained non-perturbatively, i.e.,  via numerical simulations requiring the discretization and numerical evaluation of the theory. After a brief description of relevant aspects of continuum QCD we introduce its discretization on a hyper-cubic lattice and discuss the need of iterative 
(Krylov subspace) 
methods
for the solution of the arising linear systems of equations. Due to their ill-conditioned nature,  
preconditioning is advised and the use of a domain decomposition method is discussed as a prerequisite for our multigrid construction. 

\subsection{Continuum QCD} \label{continuum_qcd}
A thorough description of QCD as a quantum field theory is beyond the scope of this paper. Instead we just introduce the reader to important concepts and the notation necessary to understand the lattice discretization. 

The degrees of freedom of QCD are matter fields, called quarks, and gauge fields, called gluons. At the heart of the numerical methods for lattice QCD is a discretized version of the continuum Dirac equation
\begin{equation}\label{Dirac_eq}
  (\Dslash+m)\psi = \eta
\end{equation} 
which describes the dynamics of the quarks and the interaction of quarks and gluons for a given gluon field background. Here, $\psi = \psi(x)$ and $\eta = \eta(x)$ represent matter fields. These depend on $x$, the points in space-time, $x=(x_0,x_1,x_2,x_3)$\footnote{Physical space-time is a four dimensional Minkowski space. We present the theory in Euclidean space-time since this version can be treated numerically. The two versions are equivalent, cf.~\cite{montvay1994quantum}.}. The gluons are represented in the Dirac operator $\Dslash$ to be discussed below, and $m$ represents a scalar mass parameter not depending on $x$. This mass parameter sets the mass of the quarks in the QCD theory.

In the continuum theory the Dirac operator $\Dslash$ can be written as
\[
  \Dslash=\sum_{\mu=0}^3\gamma_\mu \otimes \left( \partial_\mu + A_\mu \right)\,,
\]
where $ \partial_\mu = \partial / \partial x_\mu$ and $A_\mu(x)$ is the gauge field. The anti-hermitian traceless matrices $A_\mu(x)$ are elements of $\mathfrak{su}(3)$, the Lie algebra of the special unitary group $\mathrm{SU}(3)$.

The quark fields $\psi$ and $\eta$ in \eqref{Dirac_eq} carry two indices that are suppressed, i.e., $\psi=\psi_{c \sigma}$. These indices label internal degrees of freedom of the quarks; one is called color ($c=1,2,3$) and the other spin ($\sigma=0,1,2,3$). At each point $x$ in space-time, we can represent the spinor $\psi(x)$, i.e., the quark field $\psi$ at a given point $x$, as a twelve component column vector
\begin{equation}
  \psi(x)=(\psi_{10}(x),\psi_{20}(x),\psi_{30}(x),\psi_{11}(x),\ldots,\psi_{33}(x))^T\,.
\end{equation}
In case operations act unambiguously on the color but differently on the spin 
degrees of freedom we use the notation $\psi_\sigma$ to denote those components 
of the quark field belonging to the fixed spin index $\sigma$. For a given point 
$x$, $\psi_{\sigma}(x)$ is thus represented by a three component column vector 
$\psi_\sigma(x)=(\psi_{1 \sigma}(x),\psi_{2 \sigma}(x),\psi_{3 \sigma}(x))^T$. 
The value of the gauge field $A_\mu$ at point $x$ is in the matrix 
representation of $\mathfrak{su}(3)$ and acts non-trivially 
on the color and trivially on 
the spin degrees of freedom, i.e, $(A_\mu\psi)(x) = (I_4 \otimes 
A_\mu(x))\psi(x)$.

The $\gamma$-matrices $\gamma_0,\gamma_1,\gamma_2,\gamma_3 \in  \mathbb{C}^{4 
\times 4}$ act non-trivially on the spin  and trivially on the color degrees of 
freedom, i.e.\ $(\gamma_\mu\psi)(x) = (\gamma_\mu \otimes I_3)\psi(x)$ . They 
are hermitian and unitary matrices which generate a Clifford algebra,  
satisfying
\begin{equation} \label{commutativity_rel:eq}
  \gamma_\mu \gamma_\nu + \gamma_\nu \gamma_\mu = \begin{cases} 2 \cdot I_4 &\mu = \nu\\0 & \mu \neq \nu \end{cases} \quad \text{ for } \mu,\nu=0,1,2,3. 
\end{equation}
Unlike the gauge fields $A_\mu$, the $\gamma$-matrices do not depend on $x$. 
   
The covariant derivative $\partial_\mu+A_\mu$ is a ``minimal coupling extension'' of the derivative $\partial_\mu$, ensuring that $((\partial_\mu+A_\mu)\psi)(x)$ transforms in the same way as $\psi(x)$ under local gauge transformations, i.e., a local change of the coordinate system in color space. As part of the covariant derivative the $A_\mu$'s can be seen as connecting different (but infinitesimally close) space-time points. The combination of covariant derivatives and the $\gamma$-matrices ensures that $\Dslash\psi(x)$ transforms under the space-time transformations of special relativity in the same way as $\psi(x)$. Local gauge invariance and special relativity are fundamental principles of the standard model of elementary particle physics.

\subsection{Lattice QCD} \label{lattice_qcd}
In order to compute predictions in QCD from first principles and non-perturbatively, the theory of QCD has to be discretized and simulated on a computer. The discretization error is then accounted for by extrapolation to the continuum limit based on simulations at different lattice spacings. One of the most expensive tasks in these computations is the solution of the discretized Dirac equation for a given right hand side. In this section we give a brief introduction into the principles of this discretization and discuss some properties of the arising linear operators. Since the discretization is typically formulated on an equispaced lattice, this treatment of QCD is also referred to as lattice QCD. For a more detailed introduction to QCD and lattice QCD we refer the interested reader to \cite{DeGrand:2006zz, Gattringer:2010zz, montvay1994quantum}.

Consider a four-dimensional torus $\mathcal{T}$. On $\mathcal{T}$ we have a 
periodic $N_t \times N_s^3$ lattice $\mathcal{L} \subset \mathcal{T}$ with 
lattice spacing $a$ and $n_{\mathcal{L}} = N_t \cdot N_s^3$ lattice points. In 
here $N_s$  denotes the number of lattice points for each of the three space 
dimensions and $N_t$ the number of lattice points in the time dimension. Hence, 
for any $x,y \in \mathcal{L}$ there exists $p \in \mathbb{Z}^4$ such that 
$$
 y = x + a \cdot p, \quad \mbox{i.e., } \,  y_\mu = x_\mu + a \cdot p_\mu  \text{ for }  \mu=0,1,2,3.
$$
For shift operations on the lattice, we define shift vectors $ \hat{\mu} \in \mathbb{R}^4 $ by
$$
  \hat{\mu}_{\nu} = \begin{cases} a & \mu=\nu \\ 0 & \text{else.} \end{cases}
$$
A quark field $\psi$ is defined at each point of the lattice, i.e., it is a function
$$
  \begin{array}[h]{rrcl}
  \psi: & \mathcal{L} & \rightarrow & \mathbb{C}^{12} \\
        & x           & \mapsto     & \psi(x)
  \end{array}
$$
on the lattice $\mathcal{L}$ which maps a point $x \in \mathcal{L}$ to a spinor $\psi(x)$. As in continuum QCD, this spinor again has color and spin indices $\psi_{c \sigma}, \; c=1,2,3, \; \sigma=0,1,2,3$. For future use we introduce the symbols $\mathcal{C}$ and $\mathcal{S}$ for the color and the spin space, i.e.,
\[
\mathcal{C} = \{1,2,3\}, \quad \mathcal{S} = \{0,1,2,3\}.
\]

The gauge fields $A_{\mu}(x)$ connecting infinitesimally close space-time points in continuum QCD have to be replaced by objects that connect points at \emph{finite} distances. To this purpose variables $U_\mu(x)$ are introduced. $U_\mu(x)$ connects $x$ and $x+\hat{\mu}$, so that we regard $U_\mu(x)$ as being associated with the {\em link} between $x$ and $x+\hat{\mu}$. The link between $x+\hat{\mu}$ and $x$, pointing in the opposite direction, is then given by $U_\mu(x)^{-1}$. The matrices $U_\mu(x)$ are an approximation to the path-ordered exponential of the integral of $A_\mu$ along the link. They satisfy
\begin{equation*}
  U_\mu(x) \in \mathrm{SU}(3), \text{ in particular } U_\mu(x)^{-1} = U_\mu(x)^H.
\end{equation*}
Figure \ref{geometry} illustrates the naming conventions on the lattice. $U_\mu(x)$ is called a \emph{gauge link} and the set of all gauge links $\{ U_\mu(x) \, : \, x\in \mathcal{L}, \, \mu=0,1,2,3 \}$ is called \emph{configuration}.

\begin{figure}
  \begin{minipage}[t]{0.5\textwidth}
    \centering\scalebox{0.7}{\input{./figs/geometry}}
    \caption{Naming conventions on the lattice.}
    \label{geometry}
  \end{minipage}
  \hfill
  \begin{minipage}[t]{0.4\textwidth}
    \centering\scalebox{0.7}{\input{./figs/clover1}}
    \caption{The clover term.}
    \label{clover1}
  \end{minipage}
\end{figure}

The covariant derivative of the continuum theory can be discretized in many ways. Here we restrict ourselves to the widely used Wilson discretization (cf.~\cite{Wilson:1975id}), noting that the multigrid solver developed in this paper is in principle applicable to any discretization resulting in local couplings.
We define forward covariant finite differences
$$
  \left(\Delta^\mu \psi_\sigma\right)(x) = \tfrac{1}{a}\left( {U_\mu(x) \psi_\sigma(x+\hat\mu) - \psi_\sigma(x)}\right)  \stackrel{\cdot}{=} (\partial_\mu + A_\mu) \psi_\sigma(x)
$$
and backward covariant finite differences
$$
  \left(\Delta_\mu \psi_\sigma\right)(x) = \tfrac{1}{a} \left(\psi_\sigma(x) - U^H_{\mu}(x-\hat\mu) \psi_\sigma(x-\hat\mu)\right) \; .
$$
Since $( \Delta^\mu )^H = -\Delta_\mu$, the centralized covariant finite differences $(\Delta^\mu + \Delta_\mu)/2$ are anti-hermitian. The simplest discretization of the Dirac operator $\Dslash$ is then given by
$$ \textstyle
  D_N = \sum_{\mu=0}^3 \gamma_\mu \otimes \left(\Delta_\mu + \Delta^\mu\right)/2.
$$
This naive discretization generates unphysical eigenvectors, a standard phenomenon when discretizing first order derivatives using central finite differences, cf.~\cite{smith1985numerical}, also known as the ``species doubling effect'' or ``red-black instability''. The eigenspace for each eigenvalue of $D_N$ has dimension $16$, but only a one-dimensional subspace corresponds to an eigenfunction of the continuum operator. Wilson introduced the stabilization term $a\Delta_\mu \Delta^\mu$, a centralized second order covariant finite difference, to avoid this problem. The Wilson discretization of the Dirac operator is thus given by
\begin{equation}\label{WilsonDirac_eq}  \textstyle
  D_W = \frac{m_0}{a} I \;\; + \;\; \frac{1}{2} \sum_{\mu=0}^3 \Big( \gamma_\mu \otimes ( \Delta_\mu + \Delta^\mu ) - aI_4 \otimes \Delta_\mu \Delta^\mu \Big),
\end{equation}
where the mass parameter $m_0$ sets the quark mass (for further details, see~\cite{montvay1994quantum}).

The commutativity relations \eqref{commutativity_rel:eq} of the  $\gamma$-matrices imply a non-trivial symmetry of $D_W$. Defining  $\gamma_5 = \gamma_0\gamma_1\gamma_2\gamma_3$ we have $\gamma_5 \gamma_\mu = - \gamma_\mu \gamma_5$ for $\mu=0,1,2,3$, and since $\gamma_\mu$ and $\gamma_5$ are hermitian we see that $\gamma_5 \gamma_\mu$ is anti-hermitian. Thus the operator $(\gamma_5 \gamma_\mu) \otimes ( \Delta_\mu + \Delta^\mu)$ is hermitian, being the product of two anti-hermitian operators. To describe the resulting {\em $\Gamma_5$-symmetry} of the Wilson Dirac operator, we define $\Gamma_5 = I_{n_{\mathcal{L}}} \otimes \gamma_5 \otimes I_3 $ and have $(\Gamma_5 D_W)^H = \Gamma_5 D_W$.

To reduce the order of the discretization error as a function of $a$, the Sheik\-ho\-les\-la\-mi-Woh\-lert or ``clover'' term (cf.~\cite{Sheikholeslami:1985ij} and Figure \ref{clover1}), depending on a parameter $c_{sw}$, is added to the lattice Wilson Dirac operator 
\begin{equation} \label{eq:lattice_WD}
  D = D_W \;\; - \;\; \frac{c_{sw}}{32a} \sum_{\mu,\nu=0}^3 ( \gamma_\mu \gamma_\nu ) \otimes ( Q_{\mu\nu} - Q_{\nu\mu} ),
\end{equation}
where $\left(Q_{\mu\nu}\psi_\sigma\right)(x) = Q_{\mu\nu}(x) \psi_\sigma(x)$ with
$$
  \begin{array}[h]{rcl}
    Q_{\mu\nu}(x) & = & U_{\mu}(x) \, U_{\nu}(x+\hat\mu) \, U_{\mu}(x+\hat\nu)^{H} \, U_{\nu}(x)^{H} + \\
                  &   & U_{\nu}(x) \, U_{\mu}(x-\hat\mu+\hat\nu)^{H} \, U_{\nu}(x-\hat\mu)^{H} \, U_{ \mu}(x-\hat\mu) + \\
                  &   & U_{\mu}(x-\hat\mu)^{H} \, U_{\nu}(x-\hat\mu-\hat\nu)^{H} \, U_{ \mu}(x-\hat\mu-\hat\nu) \, U_{ \nu}(x-\hat\nu) + \\
                  &   & U_{\nu}(x-\hat\nu)^{H} \, U_{\mu}(x-\hat\nu) \, U_{ \nu}(x-\hat\nu+\hat\mu) \, U_{\mu}(x)^{H}. \\
  \end{array}
$$ 
The clover term is diagonal on the lattice $\mathcal{L}$. It removes $\mathcal{O}(a)$-discretization effects from the covariant finite difference discretization of the covariant derivative (for appropriately tuned $c_{\mathit{sw}}$; see \cite{Sheikholeslami:1985ij} and references therein). The resulting discretized Dirac operator $D$ thus has discretization effects of order $\mathcal{O}(a^2)$. It is again $\Gamma_5$-symmetric, i.e., we have 
\begin{equation} \label{eq:g5sym}
  (\Gamma_5 D)^H = \Gamma_5 D.
\end{equation}%
The $\Gamma_5$-symmetry induces a symmetry on the spectrum of $D$: 

\begin{lemma} \label{lem:g5sym} Every right eigenvector $\psi_\lambda$ to an eigenvalue $\lambda$ of $D$ corresponds to a left eigenvector $\hat\psi_{\bar\lambda} = \Gamma_5 \psi_{\lambda}$
to the eigenvalue $\bar\lambda$ of $D$ and vice versa. In particular, the spectrum of $D$ is symmetric with respect to the real axis.
\end{lemma}

{\em Proof.} Due to $D^H = \Gamma_5 D \Gamma_5$ we have 
\[
  D\psi_\lambda = \lambda\psi_\lambda \Leftrightarrow \psi_\lambda^H D^H = \bar{\lambda} \psi_\lambda^H \Leftrightarrow  (\Gamma_5 \psi_\lambda)^H D = \bar{\lambda} (\Gamma_5 \psi_\lambda)^H. \quad \quad \Box
\]

Summarizing, $D \in \mathbb{C}^{n\times n}$ is a sparse matrix which represents a nearest neighbor coupling on a periodic 4D lattice. The lattice has $n_{\mathcal{L}} = N_tN_s^3$ sites, each holding 12 variables, so that $n = 12n_{\mathcal{L}}$. $D$ has the symmetry property~\eqref{eq:g5sym}, depends on a mass parameter $m_0$, the Sheikholes\-la\-mi-Wohlert constant $c_{sw}$, and a configuration $\{ U_\mu(x) \, : \, x \in \mathcal{L} , \, \mu=0,1,2,3 \}$. In practice $m_0$ is negative, and for physically relevant mass parameters, the spectrum of $D$ is contained in the right half plane, cf.~Fig.~\ref{figure:spectrum} and Fig.~\ref{figure:spectrum_csw}.

\begin{figure}
  \begin{minipage}[t]{0.48\textwidth}
    \centering\scalebox{0.65}{\input{./plots/plot_spectrum1}}
    \caption{Spectrum of a $4^4$ Wilson Dirac operator; $m_0 = 0$, $c_\mathit{sw} = 0$.}
    \label{figure:spectrum}
  \end{minipage}
  \hfill
  \begin{minipage}[t]{0.48\textwidth}
    \centering\scalebox{0.65}{\input{./plots/plot_spectrum1csw}}
    \caption{Spectrum of a $4^4$ ``Clover improved'' operator; $m_0 = 0$, $c_\mathit{sw} = 1$.}
    \label{figure:spectrum_csw}
  \end{minipage}
\end{figure}

While the continuum Dirac operator is normal, the Wilson Dirac operator is not, but it approaches normality when discretization effects become smaller. For small lattice spacing, large lattice sizes and physically relevant mass parameters we can thus expect that the whole field of values $\mathcal{F}(D) = \{ \psi^HD\psi : \psi^H\psi = 1\}$ of $D$ is in the right half plane.

To explicitly formulate $D$ in matrix terms we fix a representation for the $\gamma$-matrices as
\begin{equation*}\label{eq:chiralgamma}
  \gamma_{0} = \small{
  \begin{pmatrix}
    \phantom{i} & \phantom{i} & \phantom{i} & i \\
    & & i & \phantom{i} \\
    & \makebox[0cm][c]{$-i$} & & \\
    \makebox[0cm][c]{$-i$} & & &
  \end{pmatrix}},
  \gamma_{1} = \small{ \left( \;
  \begin{matrix}
    \phantom{1} & \phantom{1} & \phantom{1} & \makebox[0cm][c]{$-1$} \\
    & & 1 & \phantom{1} \\
    & 1 & & \\
    \makebox[0cm][c]{$-1$} & & &
  \end{matrix}
  \; \right)},  
  \gamma_{2} = \small{ \left(
  \begin{matrix}
    \phantom{i} & \phantom{i} & i & \phantom{i}\\
    & & \phantom{i} & \makebox[0cm][c]{$-i$}\\
    \makebox[0cm][c]{$-i$} & & & \\
    & i & &
  \end{matrix}
  \; \right)},
  \gamma_{3} = \small{
  \begin{pmatrix}
    \phantom{1} & \phantom{1} & 1 & \phantom{1} \\
    & & \phantom{} & 1 \\
    1 & & & \\
    & 1 & &
  \end{pmatrix}
  },
\end{equation*}
resulting in
\[
\gamma_5 = \diag(1,1,-1-1).
\]
Thus $\gamma_5$ acts as the identity on spins 0 and 1 and as the negative identity on spins 2 and 3. $D$ is then given via
\begin{eqnarray*}
  (D\psi)(x) &=& \frac{1}{a} \Big( (m_0+4) I_{12} - \frac{c_{sw}}{32} \sum_{\mu,\nu=0}^3 ( \gamma_\mu \gamma_\nu ) \otimes \big( Q_{\mu\nu}(x) - Q_{\nu\mu}(x) \big) \Big) \psi(x) \\
             & &  \mbox{} - \frac{1}{2a}\sum_{\mu=0}^3 \left( (I_4-\gamma_\mu)\otimes U_\mu(x)\right) \psi(x+\hat{\mu}) \\
             & & \mbox{} - \frac{1}{2a}\sum_{\mu=0}^3 \left( (I_4+\gamma_\mu)\otimes U_\mu^H(x-\hat{\mu})\right) \psi(x-\hat{\mu}) .
\end{eqnarray*}

\subsection{Domain Decomposition in Lattice QCD} \label{dd_in_lattice_qcd}
For ease of notation we from now on drop the lattice spacing $a$, so that the lattice $\mathcal{L}$ is given as 
\[
  \mathcal{L} = \{ x = (x_0,x_1,x_2,x_3), 1 \leq x_0 \leq N_t, \, 1 \leq x_1,x_2,x_3 \leq N_s \}.
\]

Let us also reserve the terminology {\em block decomposition} for a tensor type decomposition of $\mathcal{L}$ into lattice-blocks. The precise definition is as follows.

\begin{definition}
Assume that $\{ \mathcal{T}^0_{1},\ldots, \mathcal{T}^0_{\ell_0}\}$ is a partitioning of $\{1,\ldots,N_t\}$ into $\ell_0$ blocks of consecutive time points,
\[
  \mathcal{T}^0_{j} = \{ t_{j-1}+1,\ldots,t_{j}\}, \quad j=1,\ldots,\ell_0, \, 0=t_0 < t_1 \ldots < t_{\ell_0} = N_t,
\]
and similarly for the spatial dimensions with partitionings $\{ \mathcal{T}^\mu_{1},\ldots, \mathcal{T}^\mu_{\ell_\mu}\}, \, \mu=1,2,3$.

A {\em block decomposition} of $\mathcal{L}$ is a partitioning of $\mathcal{L}$ into $\ell= \ell_0\ell_1\ell_2\ell_3$ lattice-blocks $\mathcal{L}_i$, where each lattice-block is of the form 
\[
  \mathcal{L}_i = \mathcal{T}^0_{j_0(i)} \times \mathcal{T}^1_{j_1(i)} \times \mathcal{T}^2_{j_2(i)} \times \mathcal{T}^3_{j_3(i)}.
\]
Accordingly we define a block decomposition of all $12n_\mathcal{L}$ variables in $\mathcal{V} = \mathcal{L}\times\mathcal{C}\times\mathcal{S}$ into $\ell$ blocks $\mathcal{V}_i$ by grouping all spin and color components from the lattice-block $\mathcal{L}_i$, 
\begin{equation} \label{eq:variable_blocks}
  \mathcal{V}_i = \mathcal{L}_i\times\mathcal{C}\times\mathcal{S}.
\end{equation}
\end{definition}

Since the systems arising in lattice QCD calculations tend to have hundreds of millions of unknowns they require the use of parallel computers. For this reason and due to the fact that, as a rule, naive domain decomposition is already used to parallelize the matrix vector product $Dz$ which is needed for Krylov subspace methods, it is natural to also use a domain decomposition approach as a preconditioner.

The method of choice here is a colored version of the multiplicative Schwarz method \cite{Schwarz1870, BFSmith_PEBjorstad_WDGropp_1996a}. Since the discretized Dirac operator has only nearest-neighbor couplings, only two colors are needed. 
For a block decomposition of the lattice and variable blocks $\mathcal{V}_i$ according to \eqref{eq:variable_blocks}, let 
the corresponding trivial embeddings, block systems and block solvers be denoted by
\[
  \mathcal{I}_{\mathcal{V}_i} :  \mathcal{V}_i \rightarrow \mathcal{V}, \quad D_i = \mathcal{I}_{\mathcal{V}_i}^T D \mathcal{I}_{\mathcal{V}_i} \quad \text{and} \quad B_i = \mathcal{I}_{\mathcal{V}_i} D_i^{-1} \mathcal{I}_{\mathcal{V}_i}^T.
\]

For red-black multiplicative Schwarz the lattice blocks are divided into two groups 
(red and black) such that no equation in $D$ couples variables from different 
blocks of the same color. Given the residual $r=b-Dx$, the solutions $e_i$ 
of the local systems
\begin{equation} \label{eq:localsystem}
  D_i e_i =  \mathcal{I}_{\mathcal{V}_i}^Tr,
\end{equation}
yield the corrections for the iterate $x$. More precisely,  
with the shorthand $B_{color} = \sum_{i \in \text{\it color\ }} B_{i}$ and
\[
  K = B_{\text{\it black\ }}(I - D\, B_{\text{\it red\ }}) + B_{\text{\it red\ }}
\]
we can summarize one iteration ($\nu = 1$) of red-black multiplicative Schwarz 
as (cf.~\cite{BFSmith_PEBjorstad_WDGropp_1996a})
\begin{equation} \label{SAP_1step:eq}
  z \leftarrow (I-KD)z + Kb.
\end{equation}
Since the solution $z^* = D^{-1}b$ satisfies $z^* = (I-KD)z^* + Kb$, the update for the error is
$e \leftarrow (I-KD)e$, with $I-KD$ the error propagation operator,
\begin{equation*}
  \ESAP = I - KD = (I-B_{\text{\it black\ }}D)(I-B_{\text{\it red\ }}D) \, .
\end{equation*}

The red-black Schwarz method has been introduced to lattice QCD in~\cite{Luescher2003} and has been used ever since in several lattice QCD implementations as a preconditioner (cf.~\cite{Aoki:2008sm,Nobile2012,Luescher2007}). In this context red-black Schwarz is also known as Schwarz Alternating Procedure (SAP). In what follows the application of $\nu$ iterations of SAP to a vector $b$ with initial guess $z=0$ is denoted by the linear operator
\begin{equation*}
  \MSAP{\nu} b = \sum_{k=0}^{\nu-1} (I-KD)^k \, K \, b \,.
\end{equation*}
This representation follows by repeated application of \eqref{SAP_1step:eq}. Note that $ \ESAP = I - \MSAPone D $ with $ \MSAPone = \MSAP{1} $.

\begin{figure}[ht]
  \begin{minipage}[t]{0.45\textwidth}
    \centering \scalebox{0.9}{\input{./figs/redblack_small}}
    \caption{Block decomposed lattice (reduced to $2$D) with $2$ colors}
    \label{redblack_small}
  \end{minipage}
\hfill
  \begin{minipage}[t]{0.5\textwidth}
    \centering \scalebox{0.74}{\input{plots/plot_smoother_on_spectrum.tex}}
    \caption{Error component reduction on a $4^4$ lattice with block size $2^4$}
    \label{errSAP}
  \end{minipage}
\end{figure}

Typically the solution of the local systems \eqref{eq:localsystem}, required when computing $B_ir$, is approximated by a few iterations of a Krylov subspace method (e.g., GMRES). When incorporating such an approximate solver, the SAP method becomes a non-stationary iterative process. Hence it is necessary to use flexible Krylov subspace methods like FGMRES or GCR in case that SAP is used as a preconditioner (cf.~\cite{Nobile2012, Luescher2003, Saad:2003:IMS:829576}).

It turns out that SAP as a preconditioner is not able to remedy the unfavorable 
scaling behavior of Krylov subspace methods with respect to system size, quark 
mass and physical volume. When analyzing this behavior, one realizes that SAP 
reduces error components belonging to a large part of the spectrum very well, 
but a small part is almost not affected by SAP. We illustrate this in 
Figure~\ref{errSAP} where the horizontal axis represents the eigenvectors $v$ of 
$D$ in ascending order of the absolute value of the corresponding 
eigenvalue; the vertical axis gives the ratio $\|\ESAP v\|/\|v\|$. The 
ratio is small for larger eigenvalues and becomes significantly larger for the 
small eigenvalues. This behavior is typical for a smoother in an algebraic 
multigrid method which motivated us to use SAP in this context.

\section{Algebraic Multigrid Methods} \label{section:AMG}
Any multigrid method consists of two components, a smooth\-er and a coarse grid correction \cite{Brezina2005,hackbusch2003multi,RugeStueben87,UTrottenberg_etal_2001}. Typically, the smoother can be chosen as a simple iterative method. This can be a relaxation scheme like Jacobi or Gauss-Seidel or their block variants as well as Krylov subspace methods. Given the properties of SAP presented in the previous section we choose SAP as our smoothing scheme in the QCD context.

Let us reserve the term {\em near kernel} for the space spanned by the eigenvectors belonging to small (in modulus) eigenvalues of $D$. Since SAP is not able to sufficiently remove error components belonging to the near kernel (cf.~Figure \ref{errSAP}), the multigrid method treats these persistent error components separately in a smaller subspace with $n_c$ variables. Thus, this subspace should approximate the near kernel. The typical algebraic multigrid setup is then as follows: We have to find an operator $D_c$ which resembles $D$ on that subspace both in the sense that it acts on the near kernel in a similar manner as $D$ does, but also in terms of the connection structure and sparsity. The latter allows to work on $D_c$ recursively using the same approach, thus going from two-grid to true multigrid. We also need a linear map $R: \mathbb{C}^{n} \rightarrow \mathbb{C}^{n_c}$ to restrict information from the original space to the subspace and a linear map $P: \mathbb{C}^{n_c} \rightarrow \mathbb{C}^{n}$ which transports information back to the original space. The coarse grid correction for a current iterate $z$ on the original space is then obtained by restricting the residual $r= b-Dz$ to the subspace, there solving 
\begin{equation} \label{coarse_eq}
  D_c e_c = Rr
\end{equation}
and transporting the coarse error $e_c$ back to the original space as a correction for $z$, resulting in the subspace correction
\begin{equation} \label{coarse_grid_correction}
  z \leftarrow z + P D_c^{-1} R r, \, \quad r=b-Dz
\end{equation}
with the corresponding error propagator
$$
  I - P D_c^{-1} R D.
$$

Typically, the coarse grid system is obtained as the Petrov-Galerkin projection with respect to $P$ and $R$, i.e.,
\[
  D_c = RDP.
\]
The coarse grid correction $I-P(RDP)^{-1}RD$ then is a projection onto $\range(RD)^\perp$ along $\range(P)$. The action of the coarse grid correction is thus complementary to that of the smoother if $\range(P)$ approximately contains the near kernel and $\range(RD)^\perp$ approximately contains the remaining complementary  eigenvectors (which are then efficiently reduced by the smoother). The latter condition is satisfied if $\range(R)$ approximately contains the {\em left} eigenvectors corresponding to the small eigenvalues. This can be seen by looking at exact eigenvectors: Since left and right eigenvectors are mutually orthogonal, if $\range(R)= \range(RD)$ is spanned by left eigenvectors of $D$, then $\range(R)^\perp$ is spanned by the complementary right eigenvectors of $D$.

Once $D_c$ is found a basic two-level algorithm consists of alternating the application of the smoother and the coarse grid correction. This procedure can be recursively extended to true multigrid by formulating a two-level algorithm of this kind for the solution of \eqref{coarse_eq} until we obtain an operator which is small enough to solve \eqref{coarse_eq} directly.

To be computationally beneficial, solving \eqref{coarse_eq} has to be much cheaper than solving the original equation $Dz = b$. For this purpose $D_c$ has to be very small or sparse. As the number of eigenvectors that are not sufficiently reduced by the SAP smoother grows with $n$, cf.~\cite{Banks:1979yr}, one should not aim at fixing $n_c$ (like in deflation methods), but at finding sparse matrices $R$ and $P$ whose ranges approximate the left and right near kernel of $D$ well, respectively.

\subsection{Aggregation-based Intergrid Transfer Operators} \label{sec:aggAMG}
Consider a block decomposition of the lattice $\mathcal{L}$ with lattice-blocks $\mathcal{L}_i$. It has been observed in \cite{Luescher2007} that eigenvectors belonging to small eigenvalues of $D$ tend to (approximately) coincide on a large number of lattice-blocks $\mathcal{L}_i$, a phenomenon which was termed ``local coherence''. Local coherence means in particular that we can represent many  eigenvectors with small eigenvalues from just a few by decomposing them into the parts belonging to each of the lattice-blocks. We refer to~\cite{Luescher2007} for a detailed qualitative analysis of this observation. Local coherence is the philosophy behind the aggregation-based intergrid transfer operators introduced in a general setting in \cite{DBraess_1995, Brezina2005} and applied to QCD problems in  
\cite{MGClark2010_1,MGClark2007,MGClark2010_2}.

\begin{definition} An {\em aggregation} $\{\mathcal{A}_1,\ldots,\mathcal{A}_s\}$ is a partitioning of the set $\mathcal{V} = \mathcal{L}\times\mathcal{C}\times\mathcal{S}$ of all variables. It is termed a {\em lattice-block based} aggregation if each $\mathcal{A}_i$ is of the form
\[  
  \mathcal{A}_i = \mathcal{L}_{j(i)} \times \mathcal{W}_{i},
\]
where $\mathcal{L}_j$ are the lattice-blocks of a block decomposition of $\mathcal{L}$ and $\mathcal{W}_i \subseteq \mathcal{C} \times \mathcal{S}$. 
\end{definition}

Aggregates for the lattice Wilson Dirac operator \eqref{eq:lattice_WD}  will typically be realized as lattice-block based aggregates. Note that, however, the SAP smoother on the one hand and interpolation and restriction on the other hand do not have to be based on a {\em common} block decomposition of $\mathcal{L}$.

Starting from a set of {\em test vectors} $\{ v_1,\ldots, v_N \}$ which represent the near kernel and a set of aggregates $ \{ \mathcal{A}_1,\ldots, \mathcal{A}_s \} $, the interpolation $P$ is obtained by decomposing the test vectors over the aggregates
\begin{equation}\label{eq:aggInterpolation}
  (v_1 \mid \ldots \mid v_N) = 
  \left( \begin{array}{|c|}
      \hline
      \cline{1-1}
      \multicolumn{1}{|c|}{} \\
      \multicolumn{1}{|c|}{} \\
      \multicolumn{1}{|c|}{} \\
      \multicolumn{1}{|c|}{} \\
      \multicolumn{1}{|c|}{} \\
      \multicolumn{1}{|c|}{} \\
      \multicolumn{1}{|c|}{} \\
      \hline
    \end{array} \right)
  \longrightarrow
  P = \left( \begin{array}{c c c c c c c c}
      \cline{1-1}
      \multicolumn{1}{|c|}{} & & & \\
      \multicolumn{1}{|c|}{} & & & \\
      \cline{1-2}
      & \multicolumn{1}{|c|}{} & & \\
      & \multicolumn{1}{|c|}{} & & \\
      \cline{2-2}
      & & \ddots & \\
      \cline{4-4}
      & & & \multicolumn{1}{|c|}{} \\
      & & & \multicolumn{1}{|c|}{} \\
      \cline{4-4}
    \end{array} \right)
  \begin{array}{c}
    \multirow{2}{*}{$\mathcal{A}_1$} \\ 
    \\
    \multirow{2}{*}{$\mathcal{A}_2$} \\
    \\
    \vdots \\
    \multirow{2}{*}{$\mathcal{A}_s$} \\   
    \\
  \end{array}.
\end{equation}
Formally, each aggregate $\mathcal{A}_i$ induces $N$ variables $(i-1)N+1,\ldots, iN$ on the coarse system, and we define
\begin{equation} \label{eq:aggInterpolation2}
  Pe_{(i-1)N+j} = \mathcal{I}^T_{\mathcal{A}_i}v_j, \quad i=1,\ldots,s, \, j=1,\ldots,N.
\end{equation}
Herein, $\mathcal{I}_{\mathcal{A}_i}$ represents the trivial restriction operator for the aggregate $\mathcal{A}_i$, i.e., $\mathcal{I}^T_{\mathcal{A}_i}v_j$ leaves the components of $v_j$ from $\mathcal{A}_{i}$ unchanged while zeroing all others, and $e_{(i-1)N+j}$ denotes the $(i-1)N+j$-th unit vector. For the sake of stability, the test vectors are orthonormalized locally, i.e., for each $i$ we replace $\mathcal{I}_{\mathcal{A}_i}^Tv_j$ in \eqref{eq:aggInterpolation2} by the $j$-th basis vector of an orthonormal basis of $\mbox{span}(\mathcal{I}^T_{\mathcal{A}_i}v_1,\ldots,\mathcal{I}^T_{\mathcal{A}_i}v_N)$. This does not alter the range of $P$ nor does it change the coarse grid correction operator $I-P(RDP)^{-1}RD$, and it ensures $P^HP = I$.

The restriction $R$ is obtained in an analogous manner by using a set of test vectors $\{ \hat{v}_1,\ldots, \hat{v}_N \}$ and the same aggregates to build $R^H$. Figure \ref{aggregation} illustrates a lattice-block based aggregation from a lattice point of view---again reduced to two dimensions---where in each aggregate $\mathcal{A}$ we take $\mathcal{W}_i$  as the whole set $\mathcal{C} \times \mathcal{S}$. Then the aggregates can be viewed as forming a new, coarse lattice and the sparsity and connection structure of $D_c = RDP$ resembles the one of $D$, i.e., we have again a nearest neighbor coupling. Each lattice point of the coarse grid, i.e., each aggregate, holds $N$ variables.  

\begin{figure}[ht]
  \centering \scalebox{0.7}{\input{./figs/aggregation2}}
  \caption{Aggregation-based interpolation (geometrical point of view reduced to $2$D)}
  \label{aggregation}
\end{figure}
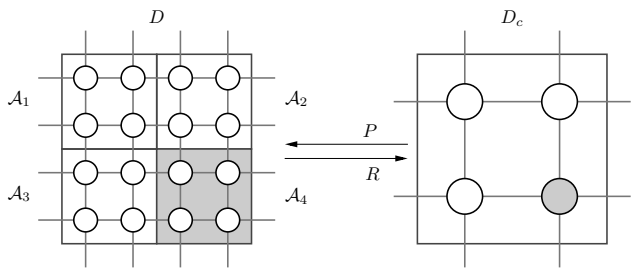

\subsection{Petrov-Galerkin Approach in Lattice QCD} \label{aggregation_based_interpolation_lattice_qcd}
The structure and the spectral properties of the Wilson Dirac operator $D$ suggest to explicitly tie the restriction $R$ to the interpolation $P$. The following construction of $P$---and thus $R$---is similar to constructions found in \cite{MGClark2010_1,MGClark2007,Luescher2007, MGClark2010_2} in the sense that the structure of all these interpolation operators is similar while the test vectors $v_i$ upon which the interpolation is built---and therefore the action of the operators---are different.

Due to Lemma~\ref{lem:g5sym} it is natural to choose
$$
  R = (\Gamma_5 P)^H
$$
in the aggregation based intergrid operators: if $P$ is built from vectors $v_1, \ldots, v_N$ which approximate right eigenvectors with small eigenvalues of $D$, then $R = (\Gamma_5 P)^H$ is built from vectors $\hat{v}_i = \Gamma_5 v_i$ which approximate left eigenvectors with small eigenvalues.

As was pointed out in \cite{MGClark2010_1}, it is furthermore possible to even obtain $R = P^{H}$ by taking the special spin-structure of the Dirac operator into account when defining the aggregates. To be specific, we introduce the following definition.
\begin{definition}
The aggregation $\{\mathcal{A}_i, \; i=1,\ldots,s\}$ is termed {\em $\Gamma_5$-compatible} if any given aggregate $\mathcal{A}_i$ is composed exclusively of fine variables with spin 0 and 1 or of fine variables with spin 2 and 3.
\end{definition}

Assume that we have a $\Gamma_5$-compatible aggregation and consider the interpolation operator $P$ from \eqref{eq:aggInterpolation}. Since $\Gamma_5$ acts as the identity on spins 0 and 1 and as the negative identity on spins 2 and 3, when going from $P$ to $\Gamma_5P$ each of the non-zero blocks in $P$ belonging to a specific aggregate is either multiplied by $+1$ or by $-1$. This gives
\begin{equation} \label{commute_with_g5}
  \Gamma_5 P = P \Gamma_5^c.
\end{equation}
with $\Gamma_5^c$ acting as the identity on the spin-0-1-aggregates and as the negative identity on the spin-2-3-aggregates.
\begin{lemma}~\label{lem:cgprops} Let the aggregation be $\Gamma_5$-compatible and $P$ the corresponding aggregation based prolongation as in \eqref{eq:aggInterpolation} and $R = (\Gamma_5P)^H$. Consider the two coarse grid operators
\[
   D_c^{PG} = RDP, \quad \mbox{ and } \quad D_c = P^HDP.
\] 
Then 
\begin{itemize}
  \item[(i)] $D_c = \Gamma_5^c D_c^{PG}$. 
  \item[(ii)] $I-PD_c^{-1}P^HD = I - P(D_c^{PG})^{-1}RD$.
  \item[(iii)] $D_c^{PG}$ is hermitian, $D_c$ is $\Gamma_5^c$-symmetric.
  \item[(iv)] For the field of values $\mathcal{F}(D)=\{\psi^{H}D\psi : \psi^{H}\psi = 1\}$, we have $\mathcal{F}(D_c) \subseteq \mathcal{F}(D)$.
\end{itemize}
\end{lemma}
\begin{proof} We first observe that just as $\Gamma_5$ the matrix $\Gamma_5^c$ is diagonal with diagonal entries $+1$ or $-1$, so $\Gamma_5^c = (\Gamma_5^c)^H = (\Gamma_5^c)^{-1}$. Part \textit{(i)} now follows from 
\[
  D_c^{PG} = RDP = (\Gamma_5P)^HDP = (P \Gamma_5^c)^HDP = \Gamma_5^c P^HDP = \Gamma_5^cD_c.
\]
Consequently, 
\[
  P(D^{PG}_c)^{-1}RD = P D_c^{-1} \Gamma_5^c P^H\Gamma_5 D = P D_c^{-1} \Gamma_5^c \Gamma_5^c P^H D  =  P D_c^{-1} P^H D,
\]
which gives \textit{(ii)}. For part \textit{(iii)} we observe that 
\[
(D_c^{PG})^H = P^HD^HR^H = P^HD^H\Gamma_5P = P^H\Gamma_5 D P = RDP = D_c^{PG}.
\]
This shows that $D_c^{PG}$ is hermitian, which is equivalent to $D_c = \Gamma_5^cD_c^{PG}$ being $\Gamma_5^c$-symmetric. Finally, since $P^HP = I$, we have 
\begin{eqnarray*}
  \mathcal{F}(D_c) \, = \, \{ \psi_c^HD_c\psi_c: \psi_c^H\psi_c = 1 \} 
                     &=& \{ (P\psi_c)^HD(P\psi_c): (P\psi_c)^H(P\psi_c) = 1 \} \\
  &\subseteq& 
  \{ \psi^HD\psi: \psi^H\psi = 1 \} = \mathcal{F}(D),
\end{eqnarray*}
which gives \textit{(iv)}.
\end{proof}

Lemma~\ref{lem:cgprops} has some remarkable consequences. Part \textit{(ii)} shows that we end up with the same coarse grid correction, irrespectively of whether we pursue a Petrov-Galerkin approach (matrix $D_c^{PG}$ with $R = \Gamma_5P$) or a Galerkin approach (matrix $D_c$, restriction is the adjoint of the prolongation). The Petrov-Galerkin matrix $D_c^{PG}$ inherits the hermiticity of the matrix $\Gamma_5D$, whereas the Galerkin matrix $D_c$ inherits the $\Gamma_5$-symmetry (and thus the symmetry of the spectrum, see Lemma~\ref{lem:g5sym}) of $D$. Moreover, if $\mathcal{F}(D)$ lies in the right half plane, then so does $\mathcal{F}(D_c)$ and thus the spectrum of $D_c$. It is known that the ``symmetrized'' Wilson Dirac operator $\Gamma_5D$ is close to maximally indefinite \cite{Gohberg_etal_2005}, i.e., the number of negative eigenvalues is about the same as the positive ones. This property is also inherited by $\Gamma_5^cD_c= D_c^{PG}$. 

$\Gamma_5$-symmetry implies an interesting connection between the eigensystem of $\Gamma_5D$ and the singular values and vectors of $D$. Indeed, if
\[
  \Gamma_5D = V \Lambda V^H, \enspace \Lambda \mbox{ diagonal}, \enspace V^HV = I 
\]
denotes the eigendecomposition of the hermitian matrix $\Gamma_5D$, then
\begin{equation} \label{eq:evd-svd}
  D = (\Gamma_5V \sign(\Lambda)) \ |\Lambda| \ V^H = U \Sigma V^H
\end{equation}
is the singular value decomposition of $D$ with the unitary matrix $U = \Gamma_5 V \sign(\Lambda)$ and $\Sigma = |\Lambda|$. 

The theory of algebraic multigrid methods for non-hermitian problems recently developed in~\cite{Sanders10} suggests to base interpolation and restriction on the right and left singular vectors corresponding to small singular values rather than on eigenvectors, so we could in principle use the relation~\eqref{eq:evd-svd}. However, obtaining good approximations for the singular vectors belonging to small singular values is now much harder than obtaining good approximations to eigenvectors belonging to small eigenvalues, since the small singular values lie right in the middle of the spectrum of $\Gamma_5D$, whereas the small eigenvalues of $D$ lie at the ``border'' of its spectrum (and in the right half plane $\mathbb{C}^+$ if $\mathcal{F}(D) \subset \mathbb{C}^+).$ Numerically we did not find that going after the singular values payed off with respect to the solver performance and it significantly increased the setup timing. These observations led us to the eigenvector based adaptive multigrid approach presented here; it also motivates that we consider $D_c$ rather than $D_c^{PG}$ as the ``correct'' coarse grid system to work with recursively in a true multigrid method.

In our computations, we take special $\Gamma_5$-compatible, lattice-block based aggregations.  

\begin{definition} \label{def:standard_aggregation} Let $\mathcal{L}_j,j=1,\ldots,s_L$ be a block decomposition of the lattice $\mathcal{L}$. Then the {\em standard aggregation} $\{\mathcal{A}_{j,\sigma}, j=1,\ldots,s_L, \sigma = 0,1\}$ is given by
\[
  \mathcal{A}_{j,0} = \mathcal{L}_j \times \{0,1\} \times \mathcal{C}, \enspace \mathcal{A}_{j,1} = \mathcal{L}_j \times \{2,3\} \times \mathcal{C}.
\]  
\end{definition}

Aggregates of the standard aggregation always combine two spin degrees of freedom in a $\Gamma_5$-compatible manner and all three color degrees of freedom. For any given $j$, the two aggregates $\mathcal{A}_{j,0}$ and $\mathcal{A}_{j,1}$ are the two only aggregates associated with the lattice-block $\mathcal{L}_j$. The standard aggregates thus induce a coarse lattice $\mathcal{L}_c$ with $n_{\mathcal{L}_c}$ sites where each coarse lattice site corresponds to one lattice-block $\mathcal{L}_j$ and holds $2N$ variables with $N$ the number of test vectors. $N$ variables correspond to spin 0 and 1 (and aggregate $\mathcal{A}_{j,0}$); another $N$ variables to spin 2 and 3 (and aggregate $\mathcal{A}_{j,1}$). Thus the overall system size of the coarse system is $n_c = 2Nn_{\mathcal{L}_c}$.
  
With standard aggregation, in addition to the properties listed in Lemma~\ref{lem:cgprops}, the coarse system $D_c = P^HDP$ also preserves the property that coarse lattice points can be arranged as a 4D periodic lattice such that the system represents a nearest neighbor coupling on this torus. Each coarse lattice point now carries $2N$ variables.  

We also note that applying $R$ and $P$ to a vector does not require any communication in a parallel implementation if whole aggregates are assigned to one process.

\subsection{Adaptivity in Aggregation-based AMG} \label{sec:aAggAMG}
If no {\em a priori} information about the near kernel is available, the test vectors $v_1,\ldots,v_N$ to be used in an aggregation based multigrid method have to be obtained computationally during a {\em setup phase}. We now briefly review the setup concept of adaptive (smoothed) aggregation as described in~\cite{Brezina2005}. We do so in the Galerkin context, i.e., we take $R = P^H$. The first fundamental idea of adaptivity in algebraic multigrid methods is to use the smoother to find error components not effectively reduced by the smoother, i.e., belonging to the near kernel. Starting with an initial random vector $u$, some iterations with the smoothing scheme on the homogeneous equations $Du = 0$ yield a vector $\tilde{v}$ rich in components that are not effectively reduced. The first set of test vectors then is the singleton $\{v\}$, and one constructs the corresponding aggregation based interpolation $P$ from \eqref{eq:aggInterpolation}. This construction guarantees that $v$ is in $\range(P)$ and thus is treated on the coarse grid. Once a first two- or multigrid method is constructed in this way, one can use it to generate an additional vector not effectively reduced by the current method by again iterating on the homogeneous system. This newly found vector is added to the set of test vectors upon which we build new interpolation and coarse grid operators. Continuing in this manner we ultimately end up with a multigrid method which converges rapidly, but possibly at a high computational cost for the setup if many vectors need to be generated and incorporated in the interpolation operator. To remedy this issue, already in~\cite{Brezina2005}, some sophisticated ideas to filter the best information out of the produced vectors, are proposed which have been partly used in the implementations of adaptive algebraic multigrid for QCD described in~\cite{MGClark2010_1,MGClark2007,MGClark2010_2}.


\subsection{Adaptivity in Bootstrap AMG} \label{sec:Boot}
It is possible to use the current multigrid method in an adaptive setup in more ways than just to test it for deficiencies by applying it to the homogeneous equation $Du = 0$. This is done in the {\em bootstrap} approach pursued in \cite{brandt2002,KahlBootstrap} which we sketch now. Details will be discussed in connection with the inexact deflation method in sections~\ref{sec:SetupID} and \ref{sec:DDML}.

The following observation is crucial: Given an eigenpair $(v_{c},\lambda_{c})$ of the generalized eigenvalue problem on the coarse grid
\begin{equation*}
  D_{c}v_{c} = \lambda_{c} P^H Pv_{c},
\end{equation*} we observe that $(Pv_c,\lambda_c)$ solves the 
constrained eigenvalue problem 
\begin{equation*}
  \text{find\ } (v,\lambda) \text{\ with\ } v \in \mbox{range}(P) \text{\ s.t.\ } P^H \left(Dv - \lambda v\right) = 0
\end{equation*} on the fine grid. 
This observation allows to use the coarse grid system as a source of information about the eigenvectors with small eigenvalues of the fine grid system. Computing eigenvectors with small $\lambda_{c}$ on the coarse grid is cheaper than on the fine grid, and applying a few iterations of the smoother to the lifted vectors $Pv_{c}$ yields useful test vectors rich in components belonging to the near kernel of the fine grid system. As we will see, the setup process used in the ``inexact deflation'' approach from~\cite{Luescher2007}, explained in the next section, can also be interpreted as a bootstrap-type setup, where instead of using an exact solution to the coarse grid eigenproblem only approximations are calculated. 

\section{Multigrid and Inexact Deflation} \label{sec:PID}
A hierarchical approach for solving the Wilson Dirac equation~\eqref{eq:discreteDirac}, which lately received attention in the lattice QCD community, was proposed in~\cite{Luescher2007}. It is a combination of what is called ``inexact deflation'' with an SAP preconditioned generalized conjugate residuals (GCR) method. The paper \cite{Luescher2007} does not relate its approach to the existing multigrid literature. The purpose of this section is to recast the formulations from \cite{Luescher2007} into established terminology from algebraic multigrid theory and to explain the limitations of the overall method from \cite{Luescher2007} which  composes its multigrid ingredients in a non-optimal manner. We also explain how the setup employed in \cite{Luescher2007} to construct the ``inexact deflation subspace'' (i.e., the test vectors) can be viewed and used as an approximate bootstrap setup in the sense of section~\ref{sec:Boot}.

\subsection{Inexact Deflation} \label{sec:ID}
The inexact deflation subspace constructed in~\cite{Luescher2007} is the range of a linear operator $P$ which resembles the definition of aggregation based interpolation from~\eqref{eq:aggInterpolation}. As in the aggregation-based construction it uses a set of test vectors $v_1, \ldots, v_{N}$ which are ``chopped'' up over aggregates (called subdomains in \cite{Luescher2007}) to obtain the locally supported columns of $P$. These aggregates are not $\Gamma_5$-compatible, so the $\Gamma_5$-symmetry is not preserved on the coarse grid operator $D_c$ which is obtained as $D_c = P^HDP$. Since the inexact deflation approach is not meant to be recursively extended to a true multilevel method, preserving important properties of the fine system on the coarse system is of lesser concern. However, within its two-level framework a (purely algebraic) deflating technique is applied when solving the coarse system.

Two projections $\pi_{L}$, $\pi_{R}$ are 
defined in \cite{Luescher2007} as follows
\begin{equation}\label{eq:pilr}
  \pi_{L} = I - DPD_{c}^{-1}P^{H}\quad \text{and}\quad \pi_{R} = I - PD_{c}^{-1}P^{H}D.
\end{equation} 
Clearly, $\pi_{R}$ is the coarse grid correction introduced in section~\ref{section:AMG}; cf.\ Lemma~\ref{lem:cgprops}(i). In the context of inexact deflation these projections and the relation $D\pi_{R} = \pi_{L}D$ are used to decompose the linear system of equations $Dz = b$ as
\begin{equation*}
  D\pi_{R} z = \pi_{L} b, \quad D(I - \pi_{R}) z = (I - \pi_{L}) b .
\end{equation*}
The second equation can be simplified to $(I-\pi_{R})z = PD_{c}^{-1}P^{H}b$. Thus the solution $z$ can be computed as $z = \pi_Rz + (I-\pi_R)z = \chi + \chi^{\prime},$ where
\begin{equation*}
  \chi^{\prime} = PD_{c}^{-1}P^{H}b
\end{equation*}
only requires the solution of the coarse grid system $D_{c}$ and 
\begin{equation*}
  D \chi = D \pi_R \chi = \pi_{L} b
\end{equation*}
is the ``inexactly deflated'' system which in~\cite{Luescher2007} is solved by a right preconditioned Krylov subspace method. To be specific, the Krylov subspace is built for the operator
\begin{equation*}
  D \pi_{R} \MSAP{\nu}
\end{equation*}
and the right hand side $\pi_{L}b$, and the Krylov subspace method is GCR (general conjugate residuals, cf.\ \cite{Saad:2003:IMS:829576}), a minimum residual approach which automatically adapts itself to the fact that the preconditioner $\MSAP{\nu}$ is not stationary, see the discussion in section~\ref{dd_in_lattice_qcd}.

\subsection{Comparison of Multigrid and Inexact Deflation} \label{sec:MGvsID}
Although the ingredients of an aggregation based algebraic multigrid method as described in section~\ref{section:AMG} and of ``inexact deflation'' as described in the previous paragraph are the same, their composition makes the difference. In the multigrid context we combine the SAP smoothing iteration with the coarse grid correction such that it gives rise to the error propagator of a V-cycle with $\nu$ post smoothing steps
\begin{equation*}
  E = (I - \MSAP{\nu} D)( I - P D_c^{-1} P^H D )\,.
\end{equation*} 
Hence we obtain for one iteration of the V-cycle
\begin{equation*}
  z \leftarrow z+C^{(\nu)}r 
\end{equation*}
where $z$ denotes the current iterate and $r$ the current residual $b-Dz$, and 
\begin{equation} \label{eq:iterationmatrix}
  C^{(\nu)} \, = \, \MSAP{\nu} + P D_c^{-1} P^{H} - \MSAP{\nu} D P D_c^{-1} P^{H} \, = \, \MSAP{\nu} \pi_L + P D_c^{-1} P^{H}\,,
\end{equation}
with the last equality following from the definition of the projectors in \eqref{eq:pilr}.
Using the multigrid method as a right preconditioner in the context of a Krylov subspace method, the preconditioner is given by $C^{(\nu)}$, and the subspace is built for $DC^{(\nu)}.$ We again should use a flexible Krylov subspace method such as flexible GMRES or GCR, since the smoother $\MSAPone$ is non-stationary and, moreover, we will solve the coarse system $D_c$ only with low accuracy using some ``inner iteration'' in every step. The important point is that a rough approximation of the coarse grid correction in \eqref{eq:iterationmatrix}, i.e., the solution of systems with the matrix $D_c$ at only low accuracy, will typically have only a negligible effect on the quality of the preconditioner, and it will certainly not hamper the convergence of the iterates towards the solution of the system since multiplications with the matrix $D$ are done exactly. On the other hand, in the ``inexact deflation'' context the exact splitting of the solution $z = \chi' + \chi$ with
\[
  \chi^{\prime} = PD_{c}^{-1}P^{H}b, \quad D \pi_R \chi = \pi_{L} b
\]
requires the same final accuracy for both $\chi'$ and $\chi$. Therefore, when computing $\chi'$, the coarse grid system  has to be solved with high accuracy. More importantly, $D_c^{-1}$ also appears in $\pi_R$ which is part of the ``deflated'' matrix $D\pi_R$ in the system for $\chi$. In the inexact deflation context, this system is solved using SAP as a preconditioner. While we can allow for a flexible and possibly inexact evaluation of the preconditioner, the accuracy with which we evaluate the non-preconditioned matrix $D \pi_R$ in every step will inevitably affect the accuracy attainable for $\chi$. As a consequence, in each iteration we have to solve the system with the matrix $D_c$ arising in $\pi_R$ with an accuracy comparable with the accuracy at which we want to obtain $\chi$ (although the accuracy requirements could, in principle, be somewhat relaxed as the iteration proceeds due to results from \cite{SiSz03,EshofSlei04}).

The difference of the two approaches is now apparent. In the multigrid context we are allowed to solve the coarse system with low accuracy, in inexact deflation we are not. Since the coarse grid system is still a large system, the work to solve it accurately will by far dominate the computational cost in each iteration in inexact deflation. In the multigrid context we can solve at only low accuracy without noticeably affecting the quality of the preconditioner, thus substantially reducing the computational cost of each iteration. Moreover, such a low accuracy solution can be obtained even more efficiently by a recursive application of the two-grid approach, resulting in a true multigrid method.
For a more detailed analysis of the connection between deflation methods (including inexact deflation) and multigrid approaches we refer to~\cite{Kahl-Rittich-preprint,Rittich2011,Tang:2010:CTP:1958286.1958296}.

\subsection{Adaptivity in the Setup of Inexact Deflation} \label{sec:SetupID}
To set up the inexact deflation method we need a way to obtain test vectors to construct the inexact deflation operators. Once these vectors are found the method is completely defined (see section~\ref{sec:ID}). In analogy to the discussion of adaptive algebraic multigrid 
in sections~\ref{sec:aAggAMG} and~\ref{sec:Boot}, these test vectors should contain information about the eigenvectors belonging to small eigenvalues of the operator $D \MSAP{\nu}$, the preconditioned system.

Though the setup proposed in~\cite{Luescher2007} is similar 
in nature 
to the one described in section~\ref{sec:aAggAMG}, it differs in one important way.  
Instead of working on the homogeneous equation $D\psi = 0$ with a random initial 
guess to obtain the test vectors, it starts with a set of random test vectors 
$\psi_j$ and approximately computes $D^{-1}\psi_j$ using SAP. The (approximate) 
multiplication with $D^{-1}$ will amplify the components of $\psi$ belonging to 
the near kernel. These new vectors are now used to define $P$ (and $D_c$), 
yielding an inexact deflation method which can again be used to approximately 
compute $D^{-1}\psi_j$ giving new vectors for $P$. The  
whole process is repeated 
several times; see Algorithm~\ref{alg:IDsetup} for a detailed description where 
a total of $\ninv$ of these cycles is performed.
\begin{algorithm}[t]
  \caption{Inexact deflation setup -- IDsetup($n_{\mathit{inv}}$,$\nu$) as used 
in \cite{Luescher2007}}\label{alg:IDsetup}
    Let $v_1,\ldots,v_{N} \in \mathbb{C}^n$  be random test vectors \;
    \For{$\eta=1$ to $3$} {
    \For{$j=1$ to $N$} {
        $v_{j} \leftarrow \MSAP{\eta} v_{j}$ 
    }
    }
    \For{$\eta=1$ to $n_{\mathit{inv}}$} { 
         (re-)construct $P$ and $D_c$ from current $v_1,\ldots,v_N$ \;
    \For{$j=1$ to $N$}{
        $v_j \leftarrow ( \MSAP{\nu} \pi_{L} + PD_c^{-1}P^{H} ) v_j$ 
\label{line:updateID} \;
     $v_j \leftarrow \frac{v_j}{||v_j||} $ \;
    }
    }
\end{algorithm}
The up\-date $v_j \leftarrow ( \MSAP{\nu} \pi_{L} + PD_c^{-1}P^{H} ) v_j$ in line~\ref{line:updateID} of the algorithm is equivalent to the application of the V-cycle iteration matrix $C^{(\nu)}$ (cf.~\eqref{eq:iterationmatrix}). It can be interpreted as one step of an iteration to solve $Dv = v_j$ with initial guess $0$ and iteration matrix $C^{(\nu)}$.

This update of the test vectors can also be viewed in terms of the bootstrap AMG setup outlined in section~\ref{sec:Boot}. While the first part of the update, $ \MSAP{\nu} \pi_{L} v_{j}$, is the application of a coarse grid correction followed by smoothing, i.e., a test to gauge the effectiveness of the method (cf.~section~\ref{sec:aAggAMG}), the second part of the update, $PD_c^{-1}P^{H} v_j$ is in $\range(P)$. In contrast to the bootstrap methodology where an update in $\range(P)$ is obtained by interpolating eigenvectors with small eigenvalues of $D_c$, in the inexact deflation variant these ``optimal'' vectors are only approximated.

\section{DD-$\alpha$AMG} \label{sec:DDML}
We now have all the ingredients available to describe our domain decomposition/aggregation based adaptive algebraic multigrid (DD-$\alpha$AMG) method for the Wilson Dirac operator~\eqref{eq:discreteDirac}. 

As its smoother we take $\MSAP{\nu}$, i.e., we perform $\nu$ iterations of red-black Schwarz as formulated in \eqref{SAP_1step:eq}. 
Like $\nu$, the underlying block decomposition of the lattice $\mathcal{L}$ is a parameter to the method which we will specify in the experiments.  

The coarse system $D_c$ is obtained as $D_c = P^HDP$, where $P$ is an aggregation based prolongation obtained during the adaptive setup phase. The aggregates are from a standard aggregation according to Definition~\ref{def:standard_aggregation}, implying that it is in particular lattice-block based and $\Gamma_5$-compatible. Parameters of the aggregation are the underlying block decomposition of $\mathcal{L}$ (which does not necessarily match the one underlying the SAP smoother) and the test vectors $v_1,\ldots,v_N$ upon which $P$ is built. The coarse grid matrix $D_c$ inherits all of the important properties of $D$, cf.\ Lemma~\ref{lem:cgprops}.

We combine the smoothing iteration and the coarse grid correction into a standard $V$-cycle with no pre- and $\nu$ steps of post smoothing so that the iteration matrix of one $V$-cycle is given by $C^{(\nu)}$ from \eqref{eq:iterationmatrix}. Instead of using iterations with the $V$-cycle as a stand-alone solver, we run FGMRES, the flexible GMRES method (cf.\ \cite{Saad:2003:IMS:829576}) with one $V$-cycle used as a (right) preconditioner. 

It remains to specify how we perform the adaptive setup yielding the test vectors $v_1,\ldots,v_N$. Extensive testing showed that a modification of the inexact deflation setup (Algorithm~\ref{alg:IDsetup}) is the most efficient. The modification is a change in the update of the vectors $v_j$ in the second half. Instead of doing one iteration with $C^{(\nu)}$ and initial guess $0$ to approximately solve $Dv=v_j$, we use the currently available vector $v_j$ as our initial guess, see Algorithm~\ref{alg:DDMLsetup}.
\begin{algorithm}[ht]
\SetKwComment{mycomment}{\hfill\{}{\}}
\SetNoFillComment
\caption{DD-$\alpha$AMG-setup($\ninv,\nu$)}\label{alg:vcycle_inv_iter}
\label{alg:DDMLsetup}
    perform Algorithm~\ref{alg:IDsetup} with line~\ref{line:updateID} 
replaced by \; 
    $v_j \leftarrow v_j + C^{(\nu)} (v_j - D v_j) $  \qquad 
\mycomment{$=C^{(\nu)} v_j + (I-C^{(\nu)}D)v_j$} 
\end{algorithm}
%
\section{Numerical Results}\label{sec:NR}
We implemented the DD-$\alpha$AMG method in the programming language \texttt{C} using the parallelization interface of \texttt{MPI}. The numerical tests focus mainly on a two-grid version of our code. As most of the work is spent on the coarse grid, a recursive extension of the two-grid method to more levels is attractive and we thus also show some results of a preliminary version of a true multigrid DD-$\alpha$AMG implementation.

Our code is optimized to a certain extent, but certainly not to the extreme. As is customary in lattice QCD computations, we use a mixed precision approach where we perform the $V$-cycle of the preconditioner in single precision. Low level optimization (e.g., making use of the SSE-registers on Intel/AMD architectures) has not been considered, yet. All Krylov subspace methods (FGMRES, BiCGStab, GCR, CG) have been implemented in a common framework with the same degree of optimization to allow  for a standardized comparison of computing times. This is particularly relevant when we compare timings with BiCGStab as well as with the multigrid method introduced in \cite{MGClark2010_1,MGClark2007,MGClark2010_2}. We also include a comparison with the inexact deflation approach, where an efficient implementation is publicly available.

A commonly used technique in lattice QCD computations is odd-even preconditioning. A lattice site $x$ is called even if $x_1+x_2+x_3+x_4$ is even, else it is called odd. Due to the nearest neighbor coupling, the Wilson Dirac operator has the form
\[
  D = \left( \begin{matrix} D_{ee} & D_{eo} \\ D_{oe} & D_{oo} \end{matrix} \right),
\]
if we order all even sites first. Herein, $D_{ee}$ and $D_{oo}$ are block diagonal with $12 \times 12$ diagonal blocks. Instead of solving a system with $D$ we can solve the corresponding system for the odd lattice sites given by the Schur complement $D_S = D_{oo}-D_{oe}D_{ee}^{-1}D_{eo}$ and then retrieve the solution at the even lattice sites, cf.\ \cite{MGClark2010_2}. The inverse $D_{ee}^{-1}$ is pre-computed once for all, and the operator $D_S$ is applied in factorized form. A matrix-vector multiplication with $D_S$ thus requires the same work as one with $D$ while the condition of $D_S$ improves over that of $D$. Typically, this results in a gain of $2-3$ in the number of iterations and execution time. Within BiCGStab we use odd-even preconditioning for $D$. In all multigrid approaches we use odd-even preconditioned restarted GMRES with a restart length of 30 when we solve the coarse system involving $D_c$. We implemented all odd-even preconditioned operators similarly in spirit to what was proposed for the Wilson Dirac operator in \cite{Krieg:2010zz}.

\begin{table}[ht]
\centering\scalebox{0.9}{\begin{tabular}{llcc}
\toprule
         & parameter                                                     &                & default         \\
\midrule 
   setup & number of iterations                                          & $\ninv$        & $6$             \\ 
         & number of test vectors                                        & $N$            & $20$            \\
         & size of lattice-blocks for aggregates                         &                & $4^4$           \\
         & coarse system relative residual tolerance                     &                & $5\cdot10^{-2}$ \\
         & (stopping criterion for the coarse system)$^{(*)}$            &                &                 \\
\midrule
  solver & restart length of FGMRES                                      & $n_{kv}$       & $25$            \\
         & relative residual tolerance  (stopping criterion)             & $\mathit{tol}$ & $10^{-10}$      \\ 
\midrule
smoother & number of post smoothing steps$^{(*)}$                        & $\nu$          & $2$             \\
         & size of lattice-blocks in SAP$^{(*)}$                         &                & $4^4$           \\
         & number of minimal residual (MR) iterations to                 &                &                 \\
         & solve the local systems \eqref{eq:localsystem} in SAP$^{(*)}$ &                & $4$             \\
\bottomrule
  \end{tabular}}
  \caption{Parameters for the DD-$\alpha$AMG two-level method.
   $(*):$ same in solver and setup}
  \label{table:allparms}
\end{table}

Table~\ref{table:allparms} summarizes the default parameters used for DD-$\alpha$AMG in our experiments. Besides those discussed in section~\ref{sec:DDML}, the table also gives the stopping criterion used for the solves with the coarse system $D_c$ (the initial residual is to be decreased by a factor of 20) and the stopping criterion for the entire FGMRES iteration (residual to be decreased by a factor of $10^{10}$). In each SAP iteration we have to (approximately) solve the local systems \eqref{eq:localsystem}. Instead of requiring a certain decrease in the residual we here fix the number of iterative steps (to $4$). The iterative method we use here is the odd-even preconditioned minimal residual method MR, i.e., restarted GMRES with a restart length of 1, where each iterative step is particularly cheap.

For the various configurations and respective matrices we found that this default set of parameters yields a well performing solver, with only little room for further tuning. The size of the lattice-blocks ($4^4$) fits well with all lattice sizes occurring in practice, where $N_t$ and $N_s$ are multiples of 4. The number of setup iterations, $\ninv$, is the only one of these parameters which should be tuned. It will depend on how many systems we have to solve, i.e., how many right hand sides we have to treat. When $\ninv$ is increased, the setup becomes more costly, while, at the same time, the solver becomes faster. Thus the time spent in the setup has to be balanced with the number of right hand sides, and we will discuss this in some detail in section~\ref{ss:setup_eval}. The default $\ninv=6$ given in Table~\ref{table:allparms} should be regarded as a good compromise.

\begin{table}[ht]
\centering\scalebox{0.9}{\begin{tabular}{ccccccc}
\toprule
  id                                    &  lattice size       & pion mass     & CGNR       & shift  & clover                 & provided by         \\
                                        &  $N_t \times N_s^3$ & $m_\pi$ [MeV] & iterations & $m_0$  & term $c_\mathit{sw}$   &                     \\
\midrule
  \ref{BMW_48_16}\conflabel{BMW_48_16}  &  $48 \times 16^3$   & $250$   & $7,\!055$  & $-0.095300$ & $1.00000$ & BMW-c~\cite{Durr:2010aw, BMW1}    \\ 
  \ref{BMW_48_24}\conflabel{BMW_48_24}  &  $48 \times 24^3$   & $250$   & $11,\!664$ & $-0.095300$ & $1.00000$ & BMW-c~\cite{Durr:2010aw, BMW1}    \\
  \ref{BMW_48_32}\conflabel{BMW_48_32}  &  $48 \times 32^3$   & $250$   & $15,\!872$ & $-0.095300$ & $1.00000$ & BMW-c~\cite{Durr:2010aw, BMW1}    \\
  \ref{BMW_48_48}\conflabel{BMW_48_48}  &  $48 \times 48^3$   & $135$   & $53,\!932$ & $-0.099330$ & $1.00000$ & BMW-c~\cite{Durr:2010aw, BMW1}    \\
  \ref{BMW_64_64}\conflabel{BMW_64_64}  &  $64 \times 64^3$   & $135$   & $84,\!207$ & $-0.052940$ & $1.00000$ & BMW-c~\cite{Durr:2010aw, BMW1}    \\
  \ref{CLS_128_64}\conflabel{CLS_128_64}&  $128 \times 64^3$  & $270$   & $45,\!804$ & $-0.342623$ & $1.75150$ & CLS~\cite{wwwCLS,Fritzsch:2012wq} \\
\bottomrule
  \end{tabular}}
  \caption{Configurations used together with their parameters. For details about their generation we refer to the references. Pion masses rounded to steps of $5$ MeV.}
\label{table:allconfs}
\end{table}

The configurations we used are listed in Table~\ref{table:allconfs}. In 
principle the pion mass $m_\pi$ and the lattice spacing (not listed) determine 
the condition of the respective matrix, e.g., the smaller $m_\pi$, the more 
ill-conditioned the respective matrix is. The physical pion mass is 
$m_{\pi_{\mathit{phys}}} = 135$ MeV which is taken on by the configurations 
\ref{BMW_48_48} and \ref{BMW_64_64}. The conditioning of the matrices is 
indicated by the iteration count of CGNR, the CG method applied to the normal 
equations $D^HD\psi = D^Hb$ (without odd-even preconditioning), in which we 
required the norm of the residual $r = b - D\psi$ to decrease by a factor of 
$10^{10}$.

We ran DD-$\alpha$AMG on the various configurations, analyzed the behavior of the setup routine and performed different scaling tests. All results have been computed on the Juropa machine at J\"ulich Supercomputing Centre, a cluster with $2,\!208$ compute nodes, each with two Intel Xeon X5570 (Nehalem-EP) quad-core processors. Unless stated otherwise the \texttt{icc}-compiler with the optimization flags \texttt{-O3 -ipo -axSSE4.2 -m64} was used.

\subsection{Comparison with BiCGStab} \label{ss:solver_comparison}
First we compare a mixed precision\footnote{The mixed precision implementation uses double precision flexible GMRES(25) preconditioned by $50$ steps of single precision, odd-even preconditioned BiCGStab}, odd-even preconditioned implementation of BiCGStab with the DD-$\alpha$AMG method using the standard parameter set for a $64^4$ configuration at physical pion mass which represents an ill-conditioned linear system with $n=201,\!326,\!592$.

\begin{table}[ht]
\centering\scalebox{0.9}{\begin{tabular}{lcccc}
\toprule
                & BiCGStab   & DD-$\alpha$AMG & speed-up factor & coarse grid     \\
\midrule 
    setup time  &            & $22.9$s        &                 &                 \\
    solve iter  & $13,\!450$ & $21$           &                 & $3,\!716^{(*)}$ \\
    solve time  & $91.2$s    & $3.15$s        & $29.0$          & $2.43$s         \\
    total time  & $91.2$s    & $26.1$s        & $3.50$          &                 \\
\bottomrule
  \end{tabular}}
  \caption{BiCGStab vs.\ DD-$\alpha$AMG with default parameters (Table~\ref{table:allparms}) on configuration~\ref{BMW_64_64} (Table~\ref{table:allconfs}), $8,\!192$ cores, $(*):$ coarse grid iterations summed up over all iterations on the fine grid.}
  \label{table:comp_64_4}
\end{table}

The results reported in Table~\ref{table:comp_64_4} show that we obtain a speed-up factor of $3.5$ over BiCGStab with respect to the total timing. Excluding the setup time, we gain a factor of $29$. The right most column 
shows that in this ill-conditioned case about 77\% of the solve time of DD-$\alpha$AMG goes into computations on the coarse grid. 

\subsection{Setup Evaluation} \label{ss:setup_eval}
Lattice QCD computations are dominated by two major tasks: generating configurations within the Hybrid Monte-Carlo (HMC) algorithm \cite{Kennedy:2006ax} and evaluating these configurations, i.e., calculating observables. Both tasks require solutions of the lattice Dirac equation.

The HMC generates a sequence of stochastically independent configurations. The configuration is changed in every step, and the Wilson Dirac equation has to be solved only once per configuration. Thus HMC requires a new setup---or at least an update---for the interpolation and coarse grid operator in every step. Therefore the costs of setup/update and solve have to be well-balanced.

The calculation of observables typically requires several solves for a single configuration. Therefore one would be willing to invest more time into the setup in order to obtain a better solver.

\begin{table}[ht]
\centering\scalebox{0.9}{\begin{tabular}{ccccccc}
\toprule
number of         & average     & average   & lowest       & highest      & average   & average   \\
setup             & setup       & iteration & iteration    & iteration    & solver    & total     \\
steps $\ninv$     & timing      & count     & count        & count        & timing    & timing    \\
\midrule
        $ 1 $     &   $2.08$    &  $ 149 $  &     $144$    &     $154$    &  $ 6.42$  & $8.50$    \\
        $ 2 $     &   $3.06$    &  $ 59.5$  &     $ 58$    &     $ 61$    &  $ 3.42$  & $6.48$    \\
        $ 3 $     &   $4.69$    &  $ 34.5$  &     $ 33$    &     $ 36$    &  $ 2.37$  & $7.06$    \\
        $ 4 $     &   $7.39$    &  $ 27.2$  &     $ 27$    &     $ 28$    &  $ 1.95$  & $9.34$    \\
        $ 5 $     &   $10.8$    &  $ 24.1$  &     $ 24$    &     $ 25$    &  $ 1.82$  & $12.6$    \\
        $ 6 $     &   $14.1$    &  $ 23.0$  &     $ 23$    &     $ 23$    &  $ 1.89$  & $16.0$    \\
        $ 8 $     &   $19.5$    &  $ 22.0$  &     $ 22$    &     $ 22$    &  $ 2.02$  & $21.5$    \\
        $10 $     &   $24.3$    &  $ 22.5$  &     $ 22$    &     $ 23$    &  $ 2.31$  & $26.6$    \\
\bottomrule
  \end{tabular}}
  \caption{Evaluation of DD-$\alpha$AMG-setup($\ninv,2$) cf.~Algorithm~\ref{alg:vcycle_inv_iter}, $48^4$ lattice, ill-conditioned configuration (Table~\ref{table:allconfs}: id~\ref{BMW_48_48}), $2,\!592$ cores, averaged over $20$ runs.}
  \label{table:setup_eval_dfl}
\end{table}

Table~\ref{table:setup_eval_dfl} illustrates how the ratio between setup and solve can be balanced depending on the amount of right hand sides. In this particular case $2$ steps in the setup might be the best choice if only a single solution of the system is needed (minimal time for setup + 1 solve). For many right hand sides, where the time spent in the solver dominates, $5$ steps in the setup might be the best choice. Doing up to $7$ steps can lower the iteration count of the solver even further, but the better the test vectors approximate the near kernel, the more ill-conditioned the coarse system becomes, i.e.,\ lowering the iteration count of the solver means increasing the iteration count on the coarse system.

The numbers shown have been averaged over $20$ runs, because the measurements vary due to the choice of random initial test vectors. The fourth and the fifth column of Table~\ref{table:setup_eval_dfl} show that the fluctuations in the iteration count of the solver are modest. For $\ninv \geq 4$ the fluctuations almost vanish completely.

\begin{table}[ht]
\centering\scalebox{0.9}{\begin{tabular}{ccccccc}
\toprule
                & \multicolumn{6}{c}{BiCGStab iteration counts}                         \\
\cmidrule(lr){2-7}
                & conf $1$  & conf $2$  & conf $3$  & conf $4$  & conf $5$  & conf $6$  \\
                & $7,\!950$ & $8,\!350$ & $9,\!550$ & $8,\!600$ & $8,\!100$ & $9,\!950$ \\
\midrule
                & \multicolumn{6}{c}{DD-$\alpha$AMG iteration counts}                   \\
\cmidrule(lr){2-7}
        $\ninv$ & conf $1$  & conf $2$  & conf $3$  & conf $4$  & conf $5$  & conf $6$  \\
        $ 1 $   &   $161$   &   $208$   &   $175$   &   $183$   &   $181$   &   $272$   \\
        $ 2 $   &   $62$    &   $75$    &   $67$    &   $67$    &   $64$    &   $85$    \\
        $ 3 $   &   $34$    &   $37$    &   $36$    &   $37$    &   $35$    &   $39$    \\
        $ 4 $   &   $27$    &   $28$    &   $28$    &   $29$    &   $27$    &   $29$    \\
        $ 5 $   &   $24$    &   $25$    &   $25$    &   $25$    &   $24$    &   $25$    \\
        $ 6 $   &   $23$    &   $23$    &   $23$    &   $24$    &   $23$    &   $23$    \\
\bottomrule
  \end{tabular}}
  \caption{Configuration dependence study of BiCGStab and DD-$\alpha$AMG with DD-$\alpha$AMG-setup($\ninv,2$) for 6 different, ill-conditioned configurations on $48^4$ lattices, (Table~\ref{table:allconfs}: id~\ref{BMW_48_48}), $2,\!592$ cores.}
  \label{table:gauge_dependence}
\end{table}

Table~\ref{table:gauge_dependence} gives the iteration count of BiCGStab and DD-$\alpha$AMG for a set of $6$ stochastically independent configurations from a single HMC simulation. The BiCGStab iteration count shows a clear dependence on the gauge fields just as DD-$\alpha$AMG for small values of $\ninv$. 
For $\ninv \geq 4$ the iteration count varies only marginally.

\subsection{Scaling Tests} \label{ss:scaling}
We now study the scaling behavior of the solver as a function of the mass parameter and the system size. While the former determines the condition number of the Wilson Dirac operator, the latter has an effect on the density of the eigenvalues. In particular, increasing the volume leads to a higher density of small eigenvalues~\cite{Banks:1979yr}. In a weak parallel scaling test we also analyze the performance as a function of the number of processors used.

\subsubsection*{Mass Scaling} \label{ss:mass_scaling}
For this study we used a $48^4$ lattice configuration. We ran the setup once for the mass parameter $m_0 = -0.09933$ in the Wilson Dirac operator~\eqref{WilsonDirac_eq}. This represents the most ill-conditioned system where the pion mass with $135$ MeV is physical. We then used the interpolation operator obtained for this system for a variety of other mass parameters, where we then ran the DD-$\alpha$AMG solver without any further setup. 

\begin{table}[ht]
\centering\scalebox{0.9}{\begin{tabular}{ccccccc}
\toprule
              & \multicolumn{2}{c}{BiCGStab} & \multicolumn{2}{c}{DD-$\alpha$AMG} & \multicolumn{2}{c}{coarse system}      \\
\cmidrule(lr){2-3}\cmidrule(lr){4-5}\cmidrule(lr){6-7}
      $m_0$   & iteration    & solver        & iteration      & solver       & \diameter iteration & timing                \\ 
              & count        & timing        & count          & timing       & count               & (\% solve time)       \\ 
\midrule
  $-0.04933$  &   $400$      &    $2.60$s    &   $17$         & $0.59$s      &         $11.2$      &   $0.13$s $(22.0)$    \\
  $-0.06933$  &   $600$      &    $4.10$s    &   $19$         & $0.72$s      &         $15.4$      &   $0.20$s $(27.8)$    \\
  $-0.08933$  &   $1,\!550$  &    $9.82$s    &   $20$         & $0.92$s      &         $28.6$      &   $0.37$s $(40.2)$    \\
  $-0.09133$  &   $1,\!700$  &    $10.6$s    &   $21$         & $1.04$s      &         $33.4$      &   $0.47$s $(45.2)$    \\
  $-0.09333$  &   $2,\!250$  &    $13.7$s    &   $21$         & $1.13$s      &         $39.7$      &   $0.55$s $(48.7)$    \\
  $-0.09533$  &   $2,\!850$  &    $17.4$s    &   $22$         & $1.28$s      &         $46.9$      &   $0.68$s $(53.1)$    \\
  $-0.09733$  &   $3,\!750$  &    $23.7$s    &   $23$         & $1.48$s      &         $56.5$      &   $0.84$s $(56.8)$    \\
  $-0.09933$  &   $6,\!250$  &    $42.0$s    &   $24$         & $1.89$s      &         $79.3$      &   $1.22$s $(64.5)$    \\
\bottomrule
  \end{tabular}}
  \caption{Mass scaling of DD-$\alpha$AMG for $\ninv=5$, $48^4$ lattice (Table~\ref{table:allconfs}: id~\ref{BMW_48_48}), $2,\!592$ cores.}
  \label{table:mass_scaling}
\end{table}

In Table~\ref{table:mass_scaling} we compare BiCGStab and DD-$\alpha$AMG with respect to the timing for one right hand side and the scaling with the mass parameter $m_0$. For the smallest $m_0$, DD-$\alpha$AMG is $22.2$ times faster than BiCGStab and even for the largest value of $m_0$ there remains a factor of $3.9$. We also see that the two methods scale in a completely different manner. The BiCGStab solve for the smallest $m_0$ is $18.5$ times more expensive than the solve for the largest one. On the other hand the DD-$\alpha$AMG timings just increase by a factor of $3.2$, the iteration count even only by a factor of $1.4$. The coarse grid iteration count, however, increases by a factor of $8.0$.

\subsubsection*{System Size Scaling} \label{ss:size_scaling}
In Table~\ref{table:lattice_scaling} we report tests on the scaling with the system size for constant mass parameter and (physical) lattice spacing. We again compare DD-$\alpha$AMG with BiCGStab. The iteration count of BiCGStab for $N_t \times N_s^3$ lattices appears to scale with $N_s$ and thus almost doubles from $N_s=16$ to $N_s=32$ whereas for DD-$\alpha$AMG we observe an almost constant iteration count and time.

\begin{table}[ht]
\centering\scalebox{0.9}{\begin{tabular}{cccccc}
\toprule
                    & \multicolumn{2}{c}{BiCGStab}    & \multicolumn{3}{c}{DD-$\alpha$AMG}        \\
\cmidrule(lr){2-3}\cmidrule(lr){4-6}
 lattice size       & iteration       & solver        & setup            & iteration & solver     \\ 
 $N_t \times N_s^3$ & count           & timing        & timing           & count     & timing     \\ 
\midrule
  $48 \times 16^3$  &   $ 1,\!550$    &    $7.03$s    &   $6.59$s        & $20$      &   $0.89$s  \\ 
  $48 \times 24^3$  &   $ 2,\!150$    &    $10.7$s    &   $7.29$s        & $20$      &   $0.83$s  \\
  $48 \times 32^3$  &   $ 2,\!600$    &    $13.1$s    &   $7.15$s        & $21$      &   $0.92$s  \\
\bottomrule
  \end{tabular}}
  \caption{Lattice size scaling of DD-$\alpha$AMG, $\ninv=6$ setup iterations, lattices generated with the same mass parameter and lattice spacing (Table~\ref{table:allconfs}: id~\ref{BMW_48_16}, \ref{BMW_48_24} and \ref{BMW_48_32}), local lattice size $4 \times 8^3$.}
  \label{table:lattice_scaling}
\end{table}

\subsubsection*{Weak Scaling} \label{ss:weak_scaling}
For a weak scaling test we ran $100$ iterations of DD-$\alpha$AMG with $\ninv = 5$ in the setup on lattices ranging from size $16^4$ on a single node ($8$ cores/node) to $128^2 \cdot 64^2$ on $1,\!024$ nodes with $16 \cdot 8^3$ local lattice size on each core, cf. Figure~\ref{plot:weak_scaling}.

\begin{figure}[tb]
  \centering\scalebox{0.65}{\input{./plots/plot_weak_scaling.tex}}
  \caption{Weak scaling test of DD-$\alpha$AMG.
    The lattice size is increased with the number of processes, keeping the local lattice size per process fixed to $16 \cdot 8^3$.}
  \label{plot:weak_scaling}
\end{figure}
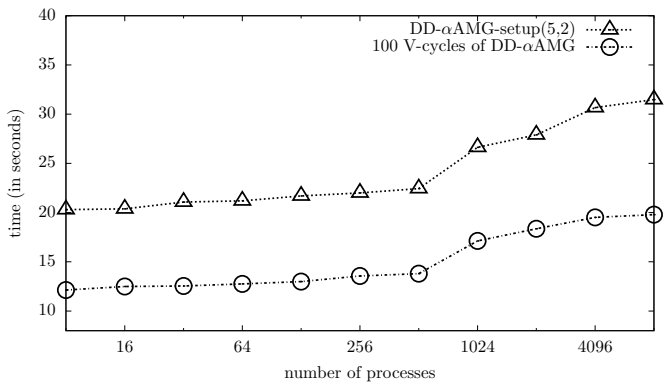

For the scaling study we fixed the number of iterations on the coarse grid to be exactly $50$ steps of odd-even preconditioned GMRES so that we always have the same number of $100$ \texttt{MPI\_Allreduce} operations. In Figure~\ref{plot:weak_scaling} we see the usual $\mathrm{log}(p)$ dependence, $p$ the number of processes, caused by global communication,
together with an exceptional increase when going from $512$ to $1,\!024$ processes. Additional measurements show that this is due to the fact that the \texttt{MPI\_Allreduce} operations take substantially longer for $1,\!024$ processors, a machine-specific feature of Juropa. Apart from this, our method scales well up to $8,\!192$ processes. 

\subsection{Comparison with the Inexact Deflation Method} \label{ss:leuscher_comp}

The inexact deflation code of \cite{Luescher2007} is publicly available \cite{wwwDDHMC}. We now compare its performance with DD-$\alpha$AMG.\footnote{Based on the preprint~\cite{Frommer:2013fsa} of the present paper, the inexact deflation method has been upgraded in the spirit of DD-$\alpha$AMG (cf.~\cite{wwwOPENQCD}). The new version is termed ``with inaccurate projection''. We here compare with the older, ``exact projection'' version.}

We have chosen the parameters of both methods equally except for the number of post-smoothing steps $\nu$. For the inexact deflation method $\nu = 5$ and for DD-$\alpha$AMG $\nu=2$ turned out to provide the fastest solver, respectively. We used the \texttt{gcc} compiler with the \texttt{-O3} flag and hand coded low-level \texttt{SSE} optimization for the inexact deflation method and the \texttt{icc} compiler with the optimization flags \texttt{-O3 -ipo -axSSE4.2 -m64} for the DD-$\alpha$AMG method. These compiler options provide the optimal choices for the respective codes. Since our focus is on algorithmic improvements we did not work on customized \texttt{SSE} optimization for DD-$\alpha$AMG, which should, in principle, give additional speed-up. 
The following results were produced on the same $48^4$ lattice as in sections~\ref{ss:setup_eval} and \ref{ss:scaling} and on a $128 \times 64^3$ lattice (Table~\ref{table:allconfs}, id~\ref{CLS_128_64}).

\begin{table}[ht]
\centering\scalebox{0.9}{\begin{tabular}{ccccccc}
\toprule
                  & \multicolumn{3}{c}{Inexact deflation}                   & \multicolumn{3}{c}{DD-$\alpha$AMG} \\
\cmidrule(lr){2-4} \cmidrule(lr){5-7}
setup             & setup         & iteration      & solver  & setup        & iteration      & solver      \\
steps $\ninv$     & timing        & count (coarse) & timing  & timing       & count (coarse) & timing      \\
\midrule
        $ 1 $     &  $1.01$s      &  $233$ $(82) $ & $10.1$s &   $2.08$s    & $149$ $(24) $  &  $ 6.42$s   \\
        $ 2 $     &  $1.87$s      &  $155$ $(145)$ & $10.2$s &   $3.06$s    & $ 59$ $(46) $  &  $ 3.42$s   \\
        $ 3 $     &  $2.69$s      &  $108$ $(224)$ & $9.96$s &   $4.69$s    & $ 35$ $(63) $  &  $ 2.37$s   \\
        $ 4 $     &  $3.43$s      &  $ 84$ $(301)$ & $9.25$s &   $7.39$s    & $ 27$ $(68) $  &  $ 1.95$s   \\
        $ 5 $     &  $6.14$s      &  $ 70$ $(320)$ & $7.50$s &   $10.8$s    & $ 24$ $(75) $  &  $ 1.82$s   \\
        $ 6 $     &  $5.68$s      &  $ 63$ $(282)$ & $5.21$s &   $14.1$s    & $ 23$ $(84) $  &  $ 1.89$s   \\
        $ 8 $     &  $7.71$s      &  $ 54$ $(267)$ & $4.12$s &   $19.5$s    & $ 22$ $(99) $  &  $ 2.02$s   \\
        $10 $     &  $10.1$s      &  $ 49$ $(265)$ & $3.62$s &   $24.3$s    & $ 22$ $(116)$  &  $ 2.31$s   \\
\bottomrule
  \end{tabular}}
  \caption{Comparison of DD-$\alpha$AMG and inexact deflation, coarse system solver tolerance $10^{-12}$ and $\nu=5$ in inexact deflation, ill-conditioned system on a $48^4$ lattice (Table~\ref{table:allconfs}: id~\ref{BMW_48_48}), $2,\!592$ cores.}
  \label{table:DDMLvsID1}
\end{table}

Table~\ref{table:DDMLvsID1} compares inexact deflation and DD-$\alpha$AMG for a whole range for $\ninv$. We see that $\ninv=5$ provides the fastest DD-$\alpha$AMG solver which is two times faster than the fastest inexact deflation solver which requires $\ninv=10$. For the calculation of observables where the same system has to be solved for several right-hand-sides, this factor of two directly carries over to the total computation time since the setup cost then is negligible.
When looking at combined times for setup and solve for one right hand side, $\ninv=2$ is best for DD-$\alpha$AMG, where it takes $6.48$s. The best choice for inexact deflation is $\ninv=6$ requiring $10.89$s.

We also see that except for very small values for $\ninv$, the number of iterations required in DD-$\alpha$AMG is less than half of that in inexact deflation. The numbers in parenthesis denote the average number of coarse solver iterations in each iteration of the respective method. For DD-$\alpha$AMG the number of iterations on the coarse grid increases with the work spent in the setup. Hence, the lowest DD-$\alpha$AMG-iteration count does not necessarily provide the fastest solver in the two grid setting. In inexact deflation the number of iterations on the coarse grid is not that clearly tied to $\ninv$. Since in inexact deflation the coarse system must be solved very accurately, the number of iterations needed to solve the coarse grid system is higher than in DD-$\alpha$AMG. It is only moderately (a factor of 2 to 4) higher, though, because the code from \cite{wwwDDHMC} uses an additional adaptively computed preconditioner for the GCR iterations on the coarse system, whereas we use the less efficient odd-even preconditioning in DD-$\alpha$AMG.

For the same number of test vectors, DD-$\alpha$AMG produces a coarse system which is twice as large (and contains four times as many non zeros in the coarse grid operator) as that of inexact deflation  with the benefit of preserving the $\Gamma_5$ structure on the coarse grid. The DD-$\alpha$AMG coarse grid system seems to be more ill-conditioned---an indication that the important aspects of the fine grid system are represented on the coarse grid---and the resulting coarse grid correction clearly lowers the total iteration count more efficiently and thus speeds up the whole method.

An ill-conditioned coarse grid system offers a potential for substantial improvements when passing from two-grid to multigrid. In Figure~\ref{plot:lvl234} we therefore report results of a study with a yet experimental version of our multigrid code. It compares the solver times for odd-even preconditioned BiCGStab with DD-$\alpha$AMG using 2, 3 and 4 levels for different choices of the mass parameter $m_0$. The gain for going from 2 to 3 levels is very noticeable; and for small values of $m_0$, corresponding to the physically interesting regimes, we observe an improvement of a factor of 2.5 to 3. Using more levels is a feature which is not available in the inexact deflation code. 

\begin{figure}[ht]
  \centering\scalebox{0.7}{\input{./plots/plot_3lvl_mass_scaling.tex}}
  \caption{Mass scaling of 2, 3 and 4 level DD-$\alpha$AMG,
  $64^4$ lattice (Table~\ref{table:allconfs}: id~\ref{BMW_64_64}), restart length $n_{kv}=10$,
  $128$ cores. Here, $m_{ud}$ is the average light quark mass (at and above which most of the current simulations are performed), assuming a common mass for the up and the down quark. Some recent simulations, e.g., \cite{Borsanyi:2013lga,Duncan:1996xy,Finkenrath:2013soa}, already distinguish the mass of the up and the down quark, and this will become more important in the near future. Then the regime close to $m_u$ becomes relevant.}
  \label{plot:lvl234}
\end{figure}
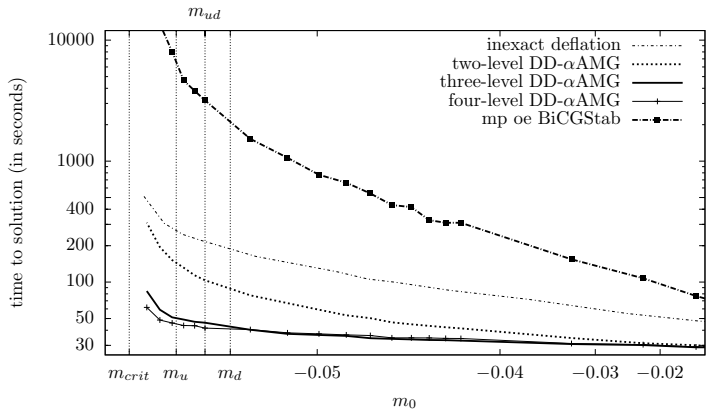 

\begin{table}[ht]
\centering\scalebox{0.9}{\begin{tabular}{lcccc}
\toprule
          & Inexact deflation & DD-$\alpha$AMG & speed-up factor \\
\midrule
    smooth iter  & $5$        & $2$            &           \\
    setup  iter  & $5$        & $3$            &           \\
    setup  time  & $10.9$s    & $7.85$s        & $1.39$    \\
    solve  iter  & $31$       & $45$           &           \\
    solve  time  & $8.63$s    & $5.81$s        & $1.49$    \\
    total  time  & $19.5$s    & $13.6$s        & $1.43$    \\
\bottomrule
  \end{tabular}}
  \caption{Comparison of DD-$\alpha$AMG with inexact deflation on an ill-conditioned system on a $128 \times 64^3$ lattice (Table~\ref{table:allconfs}, id~\ref{CLS_128_64}), same parameters as in Table~\ref{table:DDMLvsID1}, $8,\!192$ cores.}
  \label{table:comp_128_64_3}
\end{table}

Finalizing our discussion, 
we compare in Table~\ref{table:comp_128_64_3} inexact deflation and DD-$\alpha$AMG for another configuration typical for many recent lattice QCD computations. Configuration~\ref{CLS_128_64} differs from the other configurations in Table~\ref{table:allconfs} in the way it was generated, resulting in quite different discretization effects. Again we took the default parameter set, but now with relatively cheap setup phases. This results in a gain factor of more than $1.4$ for setup and solve in DD-$\alpha$AMG against inexact deflation. Still 55\% of the execution time is spent in coarse system solves  in DD-$\alpha$AMG.

\subsection{Comparison with GCR-Smoothing}
\label{ss:gmres_smoother}
The general applicability of algebraic multigrid ideas to lattice QCD systems was first considered in~\cite{MGClark2010_1, MGClark2007, MGClark2010_2} where the resulting method is simply called ``AMG'', a terminology that we keep for the following discussion. The method has been implemented as part of the QOPQDP library, see \cite{wwwQOPQDP}, which is publicly available. As we explained in sections~\ref{aggregation_based_interpolation_lattice_qcd},~\ref{sec:aAggAMG} and~\ref{sec:Boot} AMG motivated many of the choices in DD-$\alpha$AMG, particularly preservation of $\Gamma_{5}$-symmetry and aggregation based interpolation. 

There are two major differences to the method presented here. One is the choice of the smoothing iteration. While DD-$\alpha$AMG uses some steps of SAP, which can be regarded as a ``block smoothing'', AMG uses ``point smoothing'', i.e., some steps of standard (odd-even preconditioned) GCR. The other important difference is the setup. Different variants are considered in \cite{MGClark2009, MGClark2010_1, MGClark2007, MGClark2010_2}, and the QOPQDP code proceeds by computing the test vectors as quite precise approximations to eigenvectors. They are obtained one at a time by applying a sufficient number of BiCGStab iterations, at the same time keeping the current vector orthogonal to all previous ones. 

Table~\ref{table:comp_amg} reports a comparison of DD-$\alpha$AMG with AMG for two of our configurations. We compared different choices of parameters with our standard parameter settings for DD-$\alpha$AMG. We stopped the iterations when the initial residual was decreases by a factor of $10^{-5}$ (instead of $10^{-10}$), the reason for this being that in QOPQDP configurations are represented in single precision, only.  For the default choice of parameters in AMG we see that the setup is substantially more costly (factors between 2 and 4 in time), while the number of iterations for each system solve is slightly less for DD-$\alpha$AMG. We can make the effort in the AMG setup comparable to that of DD-$\alpha$AMG  by reducing the limit on the maximum number of BiCGStab iterations to be performed on each test vector ({\em msi}), but then the number of iterations for each solve increases in AMG and solve times become always larger than with DD-$\alpha$AMG. The domain decomposition smoother involves less global communication than GCR, which turns out to have a substantial influence on the solve times for a higher number of cores. For example, on $8,\!192$ cores, the solve times are 2 to 3 times smaller than in AMG.

\begin{table}[ht]
\arraycolsep0.1ex
\centering\scalebox{0.9}{\begin{tabular}{lccccccc}
  \toprule
  & \multicolumn{3}{c}{id~\ref{BMW_64_64}, $128$ cores} &\phantom{m}& \multicolumn{3}{c}{id~\ref{CLS_128_64}, $256$ cores} \\ 
                   &  AMG-d  & AMG-20  & DD-$\alpha$AMG &&   AMG-d    &  AMG-10       & DD-$\alpha$AMG                 \\
  \midrule
      setup  time  & $2424$s & $826$s  & $896$s         &&   $2464$s  &  $607$s       & $656$s                         \\
      solve  iter  & $14$    & $22$    & $10$           &&   $13$     &  $21$         & $11$                           \\
      solve  time  & $45.4$s & $66.0$s & $57.1$s        &&   $36.5$s  &  $50.4$s      & $37.3$s                        \\
  \midrule
  & \multicolumn{3}{c}{id~\ref{BMW_64_64}, $8192$ cores} && \multicolumn{3}{c}{id~\ref{CLS_128_64}, $8192$ cores}      \\ 
                   & AMG-d   & AMG-40  & DD-$\alpha$AMG &&   AMG-d    &  AMG-20       & DD-$\alpha$AMG                 \\
  \midrule
      setup  time  & $52.3$s & $24.6$s & $27.7$s        &&   $89.9$s  &  $29.1$s      & $32.3$s                        \\
      solve  iter  & $14$    & $16$    & $10$           &&   $13$     &  $16$         & $11$                           \\
      solve  time  & $4.75$s & $5.51$s & $1.82$s        &&   $3.49$s  &  $3.43$s      & $1.86$s                        \\
\bottomrule
\end{tabular}}
  \caption{Comparison of DD-$\alpha$AMG and AMG. AMG-d uses default parameter settings, AMG-$k$ sets $\textit{msi} = k$ so that setup time
is comparable to DD-$\alpha$AMG. \texttt{SSE} optimization switched off in AMG.}  
  \label{table:comp_amg}
\end{table}

\section*{Acknowledgments}
We thank two anonymous referees for several valuable suggestions. We also thank 
the Budapest-Marseille-Wuppertal collaboration for providing configurations and 
compute time on Juropa. 
We would also like to acknowledge James Brannick (Pennsylvania State University) for his advice regarding the development of the multigrid method, 
Kalman Szab{\'o} (Ber\-gi\-sche Universit\"at Wuppertal) for his support with 
implementations, Wolfgang S{\"o}ldner (University of Regensburg) for his 
help with the I/O-interfaces, and Norbert Eicker (Bergische Universit\"at Wuppertal and JSC) for sharing his expertise on Juropa.

\bibliographystyle{siam}
\bibliography{aAMG4QCD}

\end{document}

%% file: figs/geometry.tex
\begin{picture}(0,0)%
\includegraphics{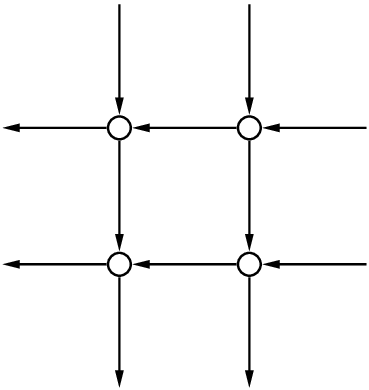}%
\end{picture}%
\setlength{\unitlength}{2279sp}%
\begingroup\makeatletter\ifx\SetFigFont\undefined%
\gdef\SetFigFont#1#2#3#4#5{%
  \reset@font\fontsize{#1}{#2pt}%
  \fontfamily{#3}\fontseries{#4}\fontshape{#5}%
  \selectfont}%
\fi\endgroup%
\begin{picture}(5106,5376)(3748,-6574)
\put(6301,-2851){\makebox(0,0)[b]{\smash{{\SetFigFont{9}{10.8}{\familydefault}{\mddefault}{\updefault}$U_\mu(x+\hat\nu)$}}}}
\put(7291,-3931){\makebox(0,0)[lb]{\smash{{\SetFigFont{9}{10.8}{\familydefault}{\mddefault}{\updefault}$U_\nu(x+\hat\mu)$}}}}
\put(4501,-4741){\makebox(0,0)[b]{\smash{{\SetFigFont{9}{10.8}{\familydefault}{\mddefault}{\updefault}$U_\mu(x-\hat\mu)$}}}}
\put(6301,-4741){\makebox(0,0)[b]{\smash{{\SetFigFont{9}{10.8}{\familydefault}{\mddefault}{\updefault}$U_\mu(x)$}}}}
\put(8101,-4741){\makebox(0,0)[b]{\smash{{\SetFigFont{9}{10.8}{\familydefault}{\mddefault}{\updefault}$U_\mu(x+\hat\mu)$}}}}
\put(5491,-5821){\makebox(0,0)[lb]{\smash{{\SetFigFont{9}{10.8}{\familydefault}{\mddefault}{\updefault}$U_\nu(x-\hat\nu)$}}}}
\put(7291,-5821){\makebox(0,0)[lb]{\smash{{\SetFigFont{9}{10.8}{\familydefault}{\mddefault}{\updefault}$U_\nu(x+\hat\mu-\hat\nu)$}}}}
\put(5491,-3211){\makebox(0,0)[lb]{\smash{{\SetFigFont{9}{10.8}{\familydefault}{\mddefault}{\updefault}$x+\hat\nu$}}}}
\put(7291,-3211){\makebox(0,0)[lb]{\smash{{\SetFigFont{9}{10.8}{\familydefault}{\mddefault}{\updefault}$x+\hat\mu+\hat\nu$}}}}
\put(5491,-1951){\makebox(0,0)[lb]{\smash{{\SetFigFont{9}{10.8}{\familydefault}{\mddefault}{\updefault}$U_\nu(x+\hat\mu)$}}}}
\put(7291,-1951){\makebox(0,0)[lb]{\smash{{\SetFigFont{9}{10.8}{\familydefault}{\mddefault}{\updefault}$U_\nu(x+\hat\mu+\hat\nu)$}}}}
\put(5491,-3931){\makebox(0,0)[lb]{\smash{{\SetFigFont{9}{10.8}{\familydefault}{\mddefault}{\updefault}$U_\nu(x)$}}}}
\put(7291,-5101){\makebox(0,0)[lb]{\smash{{\SetFigFont{9}{10.8}{\familydefault}{\mddefault}{\updefault}$x+\hat\mu$}}}}
\put(5491,-5101){\makebox(0,0)[lb]{\smash{{\SetFigFont{9}{10.8}{\familydefault}{\mddefault}{\updefault}$x$}}}}
\put(4456,-2851){\makebox(0,0)[b]{\smash{{\SetFigFont{9}{10.8}{\familydefault}{\mddefault}{\updefault}$U_\mu(x-\hat\mu+\hat\nu)$}}}}
\put(8146,-2851){\makebox(0,0)[b]{\smash{{\SetFigFont{9}{10.8}{\familydefault}{\mddefault}{\updefault}$U_\mu(x+\hat\mu+\hat\nu)$}}}}
\end{picture}%

%% file: figs/clover1.tex
\begin{picture}(0,0)%
\includegraphics{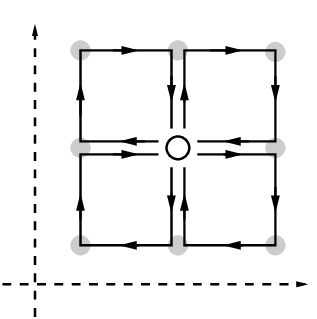}%
\end{picture}%
\setlength{\unitlength}{2279sp}%
\begingroup\makeatletter\ifx\SetFigFont\undefined%
\gdef\SetFigFont#1#2#3#4#5{%
  \reset@font\fontsize{#1}{#2pt}%
  \fontfamily{#3}\fontseries{#4}\fontshape{#5}%
  \selectfont}%
\fi\endgroup%
\begin{picture}(4368,4392)(3388,-6664)
\put(3781,-2491){\makebox(0,0)[lb]{\smash{{\SetFigFont{10}{12.0}{\familydefault}{\mddefault}{\updefault}$\hat\nu$}}}}
\put(7741,-6271){\makebox(0,0)[lb]{\smash{{\SetFigFont{10}{12.0}{\familydefault}{\mddefault}{\updefault}$\hat\mu$}}}}
\end{picture}%

%% file: plots/plot_spectrum1.tex
\begingroup
  \makeatletter
  \providecommand\color[2][]{%
    \GenericError{(gnuplot) \space\space\space\@spaces}{%
      Package color not loaded in conjunction with
      terminal option `colourtext'%
    }{See the gnuplot documentation for explanation.%
    }{Either use 'blacktext' in gnuplot or load the package
      color.sty in LaTeX.}%
    \renewcommand\color[2][]{}%
  }%
  \providecommand\includegraphics[2][]{%
    \GenericError{(gnuplot) \space\space\space\@spaces}{%
      Package graphicx or graphics not loaded%
    }{See the gnuplot documentation for explanation.%
    }{The gnuplot epslatex terminal needs graphicx.sty or graphics.sty.}%
    \renewcommand\includegraphics[2][]{}%
  }%
  \providecommand\rotatebox[2]{#2}%
  \@ifundefined{ifGPcolor}{%
    \newif\ifGPcolor
    \GPcolorfalse
  }{}%
  \@ifundefined{ifGPblacktext}{%
    \newif\ifGPblacktext
    \GPblacktexttrue
  }{}%
  \let\gplgaddtomacro\g@addto@macro
  \gdef\gplbacktext{}%
  \gdef\gplfronttext{}%
  \makeatother
  \ifGPblacktext
    \def\colorrgb#1{}%
    \def\colorgray#1{}%
  \else
    \ifGPcolor
      \def\colorrgb#1{\color[rgb]{#1}}%
      \def\colorgray#1{\color[gray]{#1}}%
      \expandafter\def\csname LTw\endcsname{\color{white}}%
      \expandafter\def\csname LTb\endcsname{\color{black}}%
      \expandafter\def\csname LTa\endcsname{\color{black}}%
      \expandafter\def\csname LT0\endcsname{\color[rgb]{1,0,0}}%
      \expandafter\def\csname LT1\endcsname{\color[rgb]{0,1,0}}%
      \expandafter\def\csname LT2\endcsname{\color[rgb]{0,0,1}}%
      \expandafter\def\csname LT3\endcsname{\color[rgb]{1,0,1}}%
      \expandafter\def\csname LT4\endcsname{\color[rgb]{0,1,1}}%
      \expandafter\def\csname LT5\endcsname{\color[rgb]{1,1,0}}%
      \expandafter\def\csname LT6\endcsname{\color[rgb]{0,0,0}}%
      \expandafter\def\csname LT7\endcsname{\color[rgb]{1,0.3,0}}%
      \expandafter\def\csname LT8\endcsname{\color[rgb]{0.5,0.5,0.5}}%
    \else
      \def\colorrgb#1{\color{black}}%
      \def\colorgray#1{\color[gray]{#1}}%
      \expandafter\def\csname LTw\endcsname{\color{white}}%
      \expandafter\def\csname LTb\endcsname{\color{black}}%
      \expandafter\def\csname LTa\endcsname{\color{black}}%
      \expandafter\def\csname LT0\endcsname{\color{black}}%
      \expandafter\def\csname LT1\endcsname{\color{black}}%
      \expandafter\def\csname LT2\endcsname{\color{black}}%
      \expandafter\def\csname LT3\endcsname{\color{black}}%
      \expandafter\def\csname LT4\endcsname{\color{black}}%
      \expandafter\def\csname LT5\endcsname{\color{black}}%
      \expandafter\def\csname LT6\endcsname{\color{black}}%
      \expandafter\def\csname LT7\endcsname{\color{black}}%
      \expandafter\def\csname LT8\endcsname{\color{black}}%
    \fi
  \fi
  \setlength{\unitlength}{0.0500bp}%
  \begin{picture}(4860.00,3402.00)%
    \gplgaddtomacro\gplbacktext{%
      \csname LTb\endcsname%
      \put(860,640){\makebox(0,0)[r]{\strut{}$-2$}}%
      \put(860,955){\makebox(0,0)[r]{\strut{}$-1.5$}}%
      \put(860,1270){\makebox(0,0)[r]{\strut{}$-1$}}%
      \put(860,1585){\makebox(0,0)[r]{\strut{}$-0.5$}}%
      \put(860,1901){\makebox(0,0)[r]{\strut{}$0$}}%
      \put(860,2216){\makebox(0,0)[r]{\strut{}$0.5$}}%
      \put(860,2531){\makebox(0,0)[r]{\strut{}$1$}}%
      \put(860,2846){\makebox(0,0)[r]{\strut{}$1.5$}}%
      \put(860,3161){\makebox(0,0)[r]{\strut{}$2$}}%
      \put(980,440){\makebox(0,0){\strut{}$0$}}%
      \put(1420,440){\makebox(0,0){\strut{}$1$}}%
      \put(1860,440){\makebox(0,0){\strut{}$2$}}%
      \put(2300,440){\makebox(0,0){\strut{}$3$}}%
      \put(2740,440){\makebox(0,0){\strut{}$4$}}%
      \put(3179,440){\makebox(0,0){\strut{}$5$}}%
      \put(3619,440){\makebox(0,0){\strut{}$6$}}%
      \put(4059,440){\makebox(0,0){\strut{}$7$}}%
      \put(4499,440){\makebox(0,0){\strut{}$8$}}%
      \put(160,1900){\rotatebox{-270}{\makebox(0,0){\strut{}imaginary axis}}}%
      \put(2739,140){\makebox(0,0){\strut{}real axis}}%
    }%
    \gplgaddtomacro\gplfronttext{%
    }%
    \gplbacktext
    \put(0,0){\includegraphics{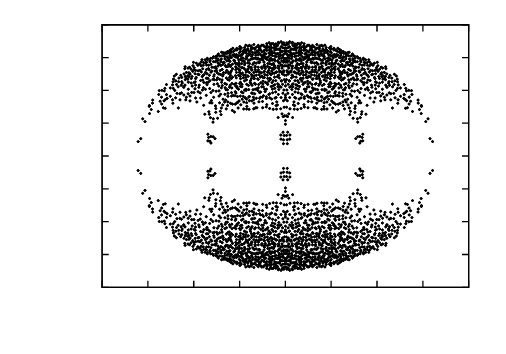}}%
    \gplfronttext
  \end{picture}%
\endgroup

%% file: plots/plot_spectrum1csw.tex
\begingroup
  \makeatletter
  \providecommand\color[2][]{%
    \GenericError{(gnuplot) \space\space\space\@spaces}{%
      Package color not loaded in conjunction with
      terminal option `colourtext'%
    }{See the gnuplot documentation for explanation.%
    }{Either use 'blacktext' in gnuplot or load the package
      color.sty in LaTeX.}%
    \renewcommand\color[2][]{}%
  }%
  \providecommand\includegraphics[2][]{%
    \GenericError{(gnuplot) \space\space\space\@spaces}{%
      Package graphicx or graphics not loaded%
    }{See the gnuplot documentation for explanation.%
    }{The gnuplot epslatex terminal needs graphicx.sty or graphics.sty.}%
    \renewcommand\includegraphics[2][]{}%
  }%
  \providecommand\rotatebox[2]{#2}%
  \@ifundefined{ifGPcolor}{%
    \newif\ifGPcolor
    \GPcolorfalse
  }{}%
  \@ifundefined{ifGPblacktext}{%
    \newif\ifGPblacktext
    \GPblacktexttrue
  }{}%
  \let\gplgaddtomacro\g@addto@macro
  \gdef\gplbacktext{}%
  \gdef\gplfronttext{}%
  \makeatother
  \ifGPblacktext
    \def\colorrgb#1{}%
    \def\colorgray#1{}%
  \else
    \ifGPcolor
      \def\colorrgb#1{\color[rgb]{#1}}%
      \def\colorgray#1{\color[gray]{#1}}%
      \expandafter\def\csname LTw\endcsname{\color{white}}%
      \expandafter\def\csname LTb\endcsname{\color{black}}%
      \expandafter\def\csname LTa\endcsname{\color{black}}%
      \expandafter\def\csname LT0\endcsname{\color[rgb]{1,0,0}}%
      \expandafter\def\csname LT1\endcsname{\color[rgb]{0,1,0}}%
      \expandafter\def\csname LT2\endcsname{\color[rgb]{0,0,1}}%
      \expandafter\def\csname LT3\endcsname{\color[rgb]{1,0,1}}%
      \expandafter\def\csname LT4\endcsname{\color[rgb]{0,1,1}}%
      \expandafter\def\csname LT5\endcsname{\color[rgb]{1,1,0}}%
      \expandafter\def\csname LT6\endcsname{\color[rgb]{0,0,0}}%
      \expandafter\def\csname LT7\endcsname{\color[rgb]{1,0.3,0}}%
      \expandafter\def\csname LT8\endcsname{\color[rgb]{0.5,0.5,0.5}}%
    \else
      \def\colorrgb#1{\color{black}}%
      \def\colorgray#1{\color[gray]{#1}}%
      \expandafter\def\csname LTw\endcsname{\color{white}}%
      \expandafter\def\csname LTb\endcsname{\color{black}}%
      \expandafter\def\csname LTa\endcsname{\color{black}}%
      \expandafter\def\csname LT0\endcsname{\color{black}}%
      \expandafter\def\csname LT1\endcsname{\color{black}}%
      \expandafter\def\csname LT2\endcsname{\color{black}}%
      \expandafter\def\csname LT3\endcsname{\color{black}}%
      \expandafter\def\csname LT4\endcsname{\color{black}}%
      \expandafter\def\csname LT5\endcsname{\color{black}}%
      \expandafter\def\csname LT6\endcsname{\color{black}}%
      \expandafter\def\csname LT7\endcsname{\color{black}}%
      \expandafter\def\csname LT8\endcsname{\color{black}}%
    \fi
  \fi
  \setlength{\unitlength}{0.0500bp}%
  \begin{picture}(4860.00,3402.00)%
    \gplgaddtomacro\gplbacktext{%
      \csname LTb\endcsname%
      \put(860,640){\makebox(0,0)[r]{\strut{}$-2$}}%
      \put(860,955){\makebox(0,0)[r]{\strut{}$-1.5$}}%
      \put(860,1270){\makebox(0,0)[r]{\strut{}$-1$}}%
      \put(860,1585){\makebox(0,0)[r]{\strut{}$-0.5$}}%
      \put(860,1901){\makebox(0,0)[r]{\strut{}$0$}}%
      \put(860,2216){\makebox(0,0)[r]{\strut{}$0.5$}}%
      \put(860,2531){\makebox(0,0)[r]{\strut{}$1$}}%
      \put(860,2846){\makebox(0,0)[r]{\strut{}$1.5$}}%
      \put(860,3161){\makebox(0,0)[r]{\strut{}$2$}}%
      \put(980,440){\makebox(0,0){\strut{}$0$}}%
      \put(1420,440){\makebox(0,0){\strut{}$1$}}%
      \put(1860,440){\makebox(0,0){\strut{}$2$}}%
      \put(2300,440){\makebox(0,0){\strut{}$3$}}%
      \put(2740,440){\makebox(0,0){\strut{}$4$}}%
      \put(3179,440){\makebox(0,0){\strut{}$5$}}%
      \put(3619,440){\makebox(0,0){\strut{}$6$}}%
      \put(4059,440){\makebox(0,0){\strut{}$7$}}%
      \put(4499,440){\makebox(0,0){\strut{}$8$}}%
      \put(160,1900){\rotatebox{-270}{\makebox(0,0){\strut{}imaginary axis}}}%
      \put(2739,140){\makebox(0,0){\strut{}real axis}}%
    }%
    \gplgaddtomacro\gplfronttext{%
    }%
    \gplbacktext
    \put(0,0){\includegraphics{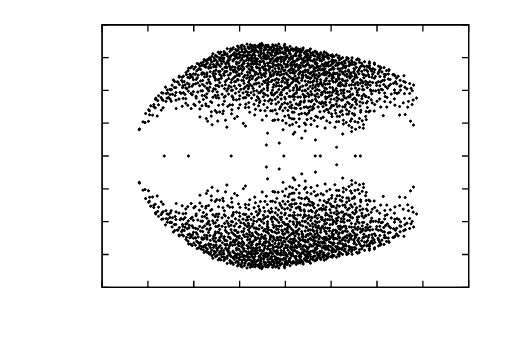}}%
    \gplfronttext
  \end{picture}%
\endgroup

%% file: figs/redblack_small.tex
\begin{picture}(0,0)%
\includegraphics{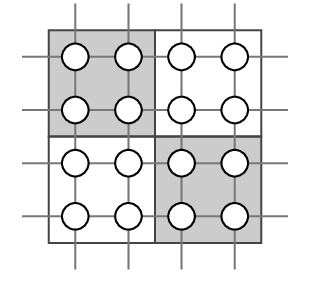}%
\end{picture}%
\setlength{\unitlength}{1865sp}%
\begingroup\makeatletter\ifx\SetFigFont\undefined%
\gdef\SetFigFont#1#2#3#4#5{%
  \reset@font\fontsize{#1}{#2pt}%
  \fontfamily{#3}\fontseries{#4}\fontshape{#5}%
  \selectfont}%
\fi\endgroup%
\begin{picture}(5295,5016)(526,-4819)
\put(541,-3031){\makebox(0,0)[b]{\smash{{\SetFigFont{10}{12.0}{\familydefault}{\mddefault}{\updefault}$\mathcal{L}_3$}}}}
\put(541,-1231){\makebox(0,0)[b]{\smash{{\SetFigFont{10}{12.0}{\familydefault}{\mddefault}{\updefault}$\mathcal{L}_1$}}}}
\put(5806,-1231){\makebox(0,0)[b]{\smash{{\SetFigFont{10}{12.0}{\familydefault}{\mddefault}{\updefault}$\mathcal{L}_2$}}}}
\put(5806,-3076){\makebox(0,0)[b]{\smash{{\SetFigFont{10}{12.0}{\familydefault}{\mddefault}{\updefault}$\mathcal{L}_4$}}}}
\put(3106,-4696){\makebox(0,0)[b]{\smash{{\SetFigFont{10}{12.0}{\familydefault}{\mddefault}{\updefault}$\mathcal{L}$}}}}
\end{picture}%

%% file: plots/plot_smoother_on_spectrum.tex
\begingroup
  \makeatletter
  \providecommand\color[2][]{%
    \GenericError{(gnuplot) \space\space\space\@spaces}{%
      Package color not loaded in conjunction with
      terminal option `colourtext'%
    }{See the gnuplot documentation for explanation.%
    }{Either use 'blacktext' in gnuplot or load the package
      color.sty in LaTeX.}%
    \renewcommand\color[2][]{}%
  }%
  \providecommand\includegraphics[2][]{%
    \GenericError{(gnuplot) \space\space\space\@spaces}{%
      Package graphicx or graphics not loaded%
    }{See the gnuplot documentation for explanation.%
    }{The gnuplot epslatex terminal needs graphicx.sty or graphics.sty.}%
    \renewcommand\includegraphics[2][]{}%
  }%
  \providecommand\rotatebox[2]{#2}%
  \@ifundefined{ifGPcolor}{%
    \newif\ifGPcolor
    \GPcolorfalse
  }{}%
  \@ifundefined{ifGPblacktext}{%
    \newif\ifGPblacktext
    \GPblacktexttrue
  }{}%
  \let\gplgaddtomacro\g@addto@macro
  \gdef\gplbacktext{}%
  \gdef\gplfronttext{}%
  \makeatother
  \ifGPblacktext
    \def\colorrgb#1{}%
    \def\colorgray#1{}%
  \else
    \ifGPcolor
      \def\colorrgb#1{\color[rgb]{#1}}%
      \def\colorgray#1{\color[gray]{#1}}%
      \expandafter\def\csname LTw\endcsname{\color{white}}%
      \expandafter\def\csname LTb\endcsname{\color{black}}%
      \expandafter\def\csname LTa\endcsname{\color{black}}%
      \expandafter\def\csname LT0\endcsname{\color[rgb]{1,0,0}}%
      \expandafter\def\csname LT1\endcsname{\color[rgb]{0,1,0}}%
      \expandafter\def\csname LT2\endcsname{\color[rgb]{0,0,1}}%
      \expandafter\def\csname LT3\endcsname{\color[rgb]{1,0,1}}%
      \expandafter\def\csname LT4\endcsname{\color[rgb]{0,1,1}}%
      \expandafter\def\csname LT5\endcsname{\color[rgb]{1,1,0}}%
      \expandafter\def\csname LT6\endcsname{\color[rgb]{0,0,0}}%
      \expandafter\def\csname LT7\endcsname{\color[rgb]{1,0.3,0}}%
      \expandafter\def\csname LT8\endcsname{\color[rgb]{0.5,0.5,0.5}}%
    \else
      \def\colorrgb#1{\color{black}}%
      \def\colorgray#1{\color[gray]{#1}}%
      \expandafter\def\csname LTw\endcsname{\color{white}}%
      \expandafter\def\csname LTb\endcsname{\color{black}}%
      \expandafter\def\csname LTa\endcsname{\color{black}}%
      \expandafter\def\csname LT0\endcsname{\color{black}}%
      \expandafter\def\csname LT1\endcsname{\color{black}}%
      \expandafter\def\csname LT2\endcsname{\color{black}}%
      \expandafter\def\csname LT3\endcsname{\color{black}}%
      \expandafter\def\csname LT4\endcsname{\color{black}}%
      \expandafter\def\csname LT5\endcsname{\color{black}}%
      \expandafter\def\csname LT6\endcsname{\color{black}}%
      \expandafter\def\csname LT7\endcsname{\color{black}}%
      \expandafter\def\csname LT8\endcsname{\color{black}}%
    \fi
  \fi
  \setlength{\unitlength}{0.0500bp}%
  \begin{picture}(4860.00,3402.00)%
    \gplgaddtomacro\gplbacktext{%
      \csname LTb\endcsname%
      \put(740,640){\makebox(0,0)[r]{\strut{}$0$}}%
      \put(740,1098){\makebox(0,0)[r]{\strut{}$0.2$}}%
      \put(740,1557){\makebox(0,0)[r]{\strut{}$0.4$}}%
      \put(740,2015){\makebox(0,0)[r]{\strut{}$0.6$}}%
      \put(740,2473){\makebox(0,0)[r]{\strut{}$0.8$}}%
      \put(740,2932){\makebox(0,0)[r]{\strut{}$1$}}%
      \put(860,440){\makebox(0,0){\strut{}$0$}}%
      \put(1447,440){\makebox(0,0){\strut{}$500$}}%
      \put(2034,440){\makebox(0,0){\strut{}$1000$}}%
      \put(2621,440){\makebox(0,0){\strut{}$1500$}}%
      \put(3208,440){\makebox(0,0){\strut{}$2000$}}%
      \put(3795,440){\makebox(0,0){\strut{}$2500$}}%
      \put(4382,440){\makebox(0,0){\strut{}$3000$}}%
      \put(160,1900){\rotatebox{-270}{\makebox(0,0){\strut{}$||E_\mathit{SAP} \, v_i||/||v_i||$}}}%
      \put(2679,140){\makebox(0,0){\strut{}number of eigenvalue}}%
    }%
    \gplgaddtomacro\gplfronttext{%
    }%
    \gplbacktext
    \put(0,0){\includegraphics{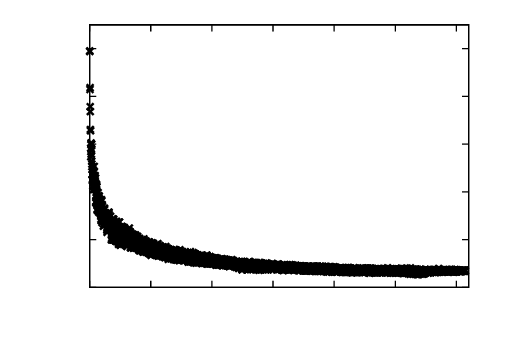}}%
    \gplfronttext
  \end{picture}%
\endgroup

%% file: figs/aggregation2.tex
\begin{picture}(0,0)%
\includegraphics{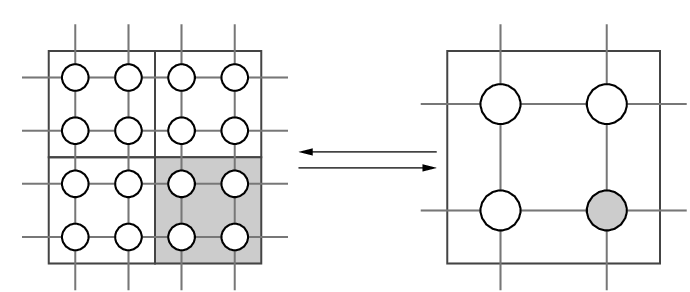}%
\end{picture}%
\setlength{\unitlength}{1865sp}%
\begingroup\makeatletter\ifx\SetFigFont\undefined%
\gdef\SetFigFont#1#2#3#4#5{%
  \reset@font\fontsize{#1}{#2pt}%
  \fontfamily{#3}\fontseries{#4}\fontshape{#5}%
  \selectfont}%
\fi\endgroup%
\begin{picture}(11658,4935)(526,-4369)
\put(541,-3031){\makebox(0,0)[b]{\smash{{\SetFigFont{10}{12.0}{\familydefault}{\mddefault}{\updefault}$\mathcal{A}_3$}}}}
\put(541,-1231){\makebox(0,0)[b]{\smash{{\SetFigFont{10}{12.0}{\familydefault}{\mddefault}{\updefault}$\mathcal{A}_1$}}}}
\put(5806,-1231){\makebox(0,0)[b]{\smash{{\SetFigFont{10}{12.0}{\familydefault}{\mddefault}{\updefault}$\mathcal{A}_2$}}}}
\put(5806,-3076){\makebox(0,0)[b]{\smash{{\SetFigFont{10}{12.0}{\familydefault}{\mddefault}{\updefault}$\mathcal{A}_4$}}}}
\put(7246,-2671){\makebox(0,0)[b]{\smash{{\SetFigFont{10}{12.0}{\familydefault}{\mddefault}{\updefault}$R$}}}}
\put(7201,-1861){\makebox(0,0)[b]{\smash{{\SetFigFont{10}{12.0}{\familydefault}{\mddefault}{\updefault}$P$}}}}
\put(3151,299){\makebox(0,0)[b]{\smash{{\SetFigFont{10}{12.0}{\familydefault}{\mddefault}{\updefault}$D$}}}}
\put(9901,299){\makebox(0,0)[b]{\smash{{\SetFigFont{10}{12.0}{\familydefault}{\mddefault}{\updefault}$D_c$}}}}
\end{picture}%

%% file: plots/plot_weak_scaling.tex
\begingroup
  \makeatletter
  \providecommand\color[2][]{%
    \GenericError{(gnuplot) \space\space\space\@spaces}{%
      Package color not loaded in conjunction with
      terminal option `colourtext'%
    }{See the gnuplot documentation for explanation.%
    }{Either use 'blacktext' in gnuplot or load the package
      color.sty in LaTeX.}%
    \renewcommand\color[2][]{}%
  }%
  \providecommand\includegraphics[2][]{%
    \GenericError{(gnuplot) \space\space\space\@spaces}{%
      Package graphicx or graphics not loaded%
    }{See the gnuplot documentation for explanation.%
    }{The gnuplot epslatex terminal needs graphicx.sty or graphics.sty.}%
    \renewcommand\includegraphics[2][]{}%
  }%
  \providecommand\rotatebox[2]{#2}%
  \@ifundefined{ifGPcolor}{%
    \newif\ifGPcolor
    \GPcolorfalse
  }{}%
  \@ifundefined{ifGPblacktext}{%
    \newif\ifGPblacktext
    \GPblacktexttrue
  }{}%
  \let\gplgaddtomacro\g@addto@macro
  \gdef\gplbacktext{}%
  \gdef\gplfronttext{}%
  \makeatother
  \ifGPblacktext
    \def\colorrgb#1{}%
    \def\colorgray#1{}%
  \else
    \ifGPcolor
      \def\colorrgb#1{\color[rgb]{#1}}%
      \def\colorgray#1{\color[gray]{#1}}%
      \expandafter\def\csname LTw\endcsname{\color{white}}%
      \expandafter\def\csname LTb\endcsname{\color{black}}%
      \expandafter\def\csname LTa\endcsname{\color{black}}%
      \expandafter\def\csname LT0\endcsname{\color[rgb]{1,0,0}}%
      \expandafter\def\csname LT1\endcsname{\color[rgb]{0,1,0}}%
      \expandafter\def\csname LT2\endcsname{\color[rgb]{0,0,1}}%
      \expandafter\def\csname LT3\endcsname{\color[rgb]{1,0,1}}%
      \expandafter\def\csname LT4\endcsname{\color[rgb]{0,1,1}}%
      \expandafter\def\csname LT5\endcsname{\color[rgb]{1,1,0}}%
      \expandafter\def\csname LT6\endcsname{\color[rgb]{0,0,0}}%
      \expandafter\def\csname LT7\endcsname{\color[rgb]{1,0.3,0}}%
      \expandafter\def\csname LT8\endcsname{\color[rgb]{0.5,0.5,0.5}}%
    \else
      \def\colorrgb#1{\color{black}}%
      \def\colorgray#1{\color[gray]{#1}}%
      \expandafter\def\csname LTw\endcsname{\color{white}}%
      \expandafter\def\csname LTb\endcsname{\color{black}}%
      \expandafter\def\csname LTa\endcsname{\color{black}}%
      \expandafter\def\csname LT0\endcsname{\color{black}}%
      \expandafter\def\csname LT1\endcsname{\color{black}}%
      \expandafter\def\csname LT2\endcsname{\color{black}}%
      \expandafter\def\csname LT3\endcsname{\color{black}}%
      \expandafter\def\csname LT4\endcsname{\color{black}}%
      \expandafter\def\csname LT5\endcsname{\color{black}}%
      \expandafter\def\csname LT6\endcsname{\color{black}}%
      \expandafter\def\csname LT7\endcsname{\color{black}}%
      \expandafter\def\csname LT8\endcsname{\color{black}}%
    \fi
  \fi
  \setlength{\unitlength}{0.0500bp}%
  \begin{picture}(7920.00,4536.00)%
    \gplgaddtomacro\gplbacktext{%
      \csname LTb\endcsname%
      \put(620,868){\makebox(0,0)[r]{\strut{}$10$}}%
      \put(620,1440){\makebox(0,0)[r]{\strut{}$15$}}%
      \put(620,2011){\makebox(0,0)[r]{\strut{}$20$}}%
      \put(620,2582){\makebox(0,0)[r]{\strut{}$25$}}%
      \put(620,3153){\makebox(0,0)[r]{\strut{}$30$}}%
      \put(620,3724){\makebox(0,0)[r]{\strut{}$35$}}%
      \put(620,4295){\makebox(0,0)[r]{\strut{}$40$}}%
      \put(1422,440){\makebox(0,0){\strut{}$16$}}%
      \put(2785,440){\makebox(0,0){\strut{}$64$}}%
      \put(4149,440){\makebox(0,0){\strut{}$256$}}%
      \put(5513,440){\makebox(0,0){\strut{}$1024$}}%
      \put(6876,440){\makebox(0,0){\strut{}$4096$}}%
      \put(160,2467){\rotatebox{-270}{\makebox(0,0){\strut{}time (in seconds)}}}%
      \put(4149,140){\makebox(0,0){\strut{}number of processes}}%
    }%
    \gplgaddtomacro\gplfronttext{%
      \csname LTb\endcsname%
      \put(6655,4132){\makebox(0,0)[r]{\strut{}DD-$\alpha$AMG-setup(5,2)}}%
      \csname LTb\endcsname%
      \put(6655,3932){\makebox(0,0)[r]{\strut{}100 V-cycles of DD-$\alpha$AMG}}%
    }%
    \gplbacktext
    \put(0,0){\includegraphics{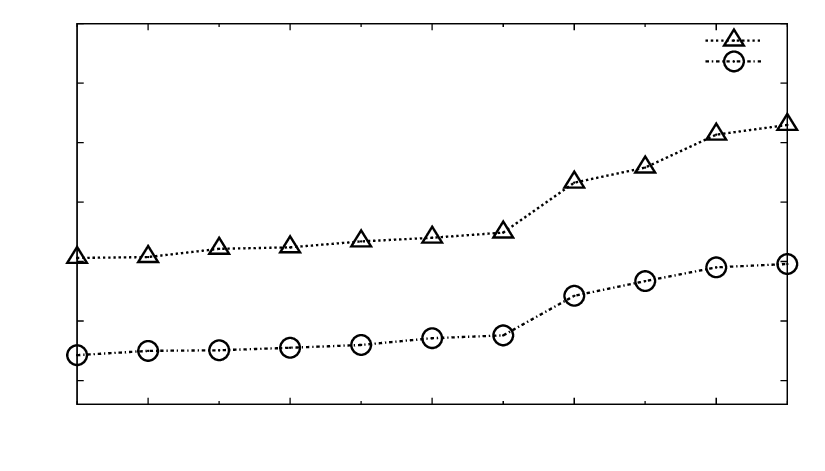}}%
    \gplfronttext
  \end{picture}%
\endgroup

%% file: plots/plot_3lvl_mass_scaling.tex
\begingroup
  \makeatletter
  \providecommand\color[2][]{%
    \GenericError{(gnuplot) \space\space\space\@spaces}{%
      Package color not loaded in conjunction with
      terminal option `colourtext'%
    }{See the gnuplot documentation for explanation.%
    }{Either use 'blacktext' in gnuplot or load the package
      color.sty in LaTeX.}%
    \renewcommand\color[2][]{}%
  }%
  \providecommand\includegraphics[2][]{%
    \GenericError{(gnuplot) \space\space\space\@spaces}{%
      Package graphicx or graphics not loaded%
    }{See the gnuplot documentation for explanation.%
    }{The gnuplot epslatex terminal needs graphicx.sty or graphics.sty.}%
    \renewcommand\includegraphics[2][]{}%
  }%
  \providecommand\rotatebox[2]{#2}%
  \@ifundefined{ifGPcolor}{%
    \newif\ifGPcolor
    \GPcolorfalse
  }{}%
  \@ifundefined{ifGPblacktext}{%
    \newif\ifGPblacktext
    \GPblacktexttrue
  }{}%
  \let\gplgaddtomacro\g@addto@macro
  \gdef\gplbacktext{}%
  \gdef\gplfronttext{}%
  \makeatother
  \ifGPblacktext
    \def\colorrgb#1{}%
    \def\colorgray#1{}%
  \else
    \ifGPcolor
      \def\colorrgb#1{\color[rgb]{#1}}%
      \def\colorgray#1{\color[gray]{#1}}%
      \expandafter\def\csname LTw\endcsname{\color{white}}%
      \expandafter\def\csname LTb\endcsname{\color{black}}%
      \expandafter\def\csname LTa\endcsname{\color{black}}%
      \expandafter\def\csname LT0\endcsname{\color[rgb]{1,0,0}}%
      \expandafter\def\csname LT1\endcsname{\color[rgb]{0,1,0}}%
      \expandafter\def\csname LT2\endcsname{\color[rgb]{0,0,1}}%
      \expandafter\def\csname LT3\endcsname{\color[rgb]{1,0,1}}%
      \expandafter\def\csname LT4\endcsname{\color[rgb]{0,1,1}}%
      \expandafter\def\csname LT5\endcsname{\color[rgb]{1,1,0}}%
      \expandafter\def\csname LT6\endcsname{\color[rgb]{0,0,0}}%
      \expandafter\def\csname LT7\endcsname{\color[rgb]{1,0.3,0}}%
      \expandafter\def\csname LT8\endcsname{\color[rgb]{0.5,0.5,0.5}}%
    \else
      \def\colorrgb#1{\color{black}}%
      \def\colorgray#1{\color[gray]{#1}}%
      \expandafter\def\csname LTw\endcsname{\color{white}}%
      \expandafter\def\csname LTb\endcsname{\color{black}}%
      \expandafter\def\csname LTa\endcsname{\color{black}}%
      \expandafter\def\csname LT0\endcsname{\color{black}}%
      \expandafter\def\csname LT1\endcsname{\color{black}}%
      \expandafter\def\csname LT2\endcsname{\color{black}}%
      \expandafter\def\csname LT3\endcsname{\color{black}}%
      \expandafter\def\csname LT4\endcsname{\color{black}}%
      \expandafter\def\csname LT5\endcsname{\color{black}}%
      \expandafter\def\csname LT6\endcsname{\color{black}}%
      \expandafter\def\csname LT7\endcsname{\color{black}}%
      \expandafter\def\csname LT8\endcsname{\color{black}}%
    \fi
  \fi
  \setlength{\unitlength}{0.0500bp}%
  \begin{picture}(7920.00,4536.00)%
    \gplgaddtomacro\gplbacktext{%
      \csname LTb\endcsname%
      \put(980,743){\makebox(0,0)[r]{\strut{}$30$}}%
      \put(980,1032){\makebox(0,0)[r]{\strut{}$50$}}%
      \put(980,1817){\makebox(0,0)[r]{\strut{}$200$}}%
      \put(980,2047){\makebox(0,0)[r]{\strut{} }}%
      \put(980,2210){\makebox(0,0)[r]{\strut{}$400$}}%
      \put(980,2336){\makebox(0,0)[r]{\strut{} }}%
      \put(980,1425){\makebox(0,0)[r]{\strut{}$100$}}%
      \put(980,2728){\makebox(0,0)[r]{\strut{}$1000$}}%
      \put(980,4032){\makebox(0,0)[r]{\strut{}$10000$}}%
      \put(7075,440){\makebox(0,0){\strut{}$-0.02$}}%
      \put(6377,440){\makebox(0,0){\strut{}$-0.03$}}%
      \put(5351,440){\makebox(0,0){\strut{}$-0.04$}}%
      \put(3383,440){\makebox(0,0){\strut{}$-0.05$}}%
      \put(2447,440){\makebox(0,0){\strut{}$m_d$}}%
      \put(2175,440){\makebox(0,0){\strut{} }}%
      \put(1865,440){\makebox(0,0){\strut{}$m_u$}}%
      \put(1359,440){\makebox(0,0){\strut{}$m_{crit}$}}%
      \put(2175,4335){\makebox(0,0){\strut{}$m_{ud}$}}%
      \put(160,2387){\rotatebox{-270}{\makebox(0,0){\strut{}time to solution (in seconds)}}}%
      \put(4329,140){\makebox(0,0){\strut{}$m_0$}}%
    }%
    \gplgaddtomacro\gplfronttext{%
      \csname LTb\endcsname%
      \put(6655,3972){\makebox(0,0)[r]{\strut{}inexact deflation}}%
      \csname LTb\endcsname%
      \put(6655,3772){\makebox(0,0)[r]{\strut{}two-level DD-$\alpha$AMG}}%
      \csname LTb\endcsname%
      \put(6655,3572){\makebox(0,0)[r]{\strut{}three-level DD-$\alpha$AMG}}%
      \csname LTb\endcsname%
      \put(6655,3372){\makebox(0,0)[r]{\strut{}four-level DD-$\alpha$AMG}}%
      \csname LTb\endcsname%
      \put(6655,3172){\makebox(0,0)[r]{\strut{}mp oe BiCGStab}}%
    }%
    \gplbacktext
    \put(0,0){\includegraphics{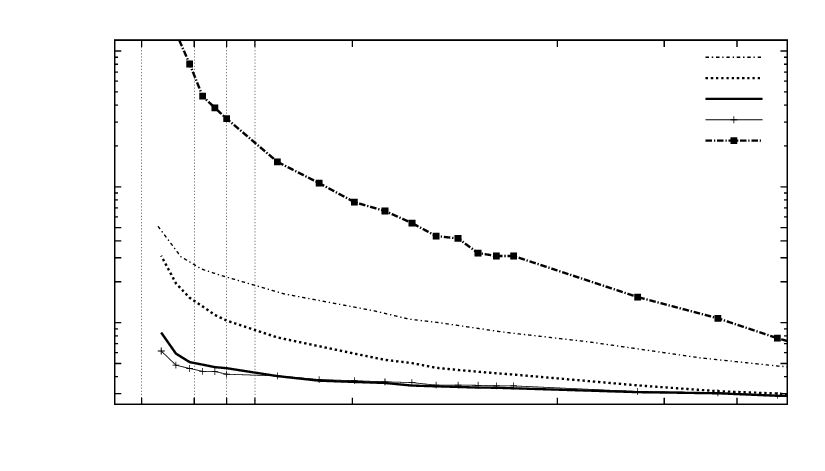}}%
    \gplfronttext
  \end{picture}%
\endgroup